\documentclass[11pt,a4paper]{article}

\newcommand*{\images}{images}

\usepackage{import}
\usepackage{archive}
\usepackage[T1]{fontenc}

\begin{document}

\author[1]{Katherine Molinet}
\author[1]{Aris Filos-Ratsikas}

\affil[1]{University of Edinburgh}

\date{}

\title{A Theoretical Approach to Stablecoin Design via Price Windows\thanks{\protect This work was supported by Input Output (\href{https://www.iog.io/}{iog.io}) through their funding of the Edinburgh Blockchain Technology Lab (BTL) and the IO Research Hub (IORH). This work was also supported by the UK Engineering and Physical Sciences Research Council (EPSRC) grant EP/Y003624/1. The authors would also like to thank Aggelos Kiayias for useful discussions during the initial stages of this paper.}}

% Citing papers: ``\citet{klages2020stablecoins} did this'', or ``as it can be seen in \citep{klages2020stablecoins}'' etc.  
% Referencing sections: \cref{sec:introduction} and in the appendix: \cref{app:section1}.

\maketitle

\begin{abstract}
In this paper, we explore the short- and long-term stability of backed stablecoins offering constant mint and redeem prices to all agents. We refer to such designs as \textit{price window-based,} since the mint and redeem prices constrain the stablecoin's market equilibrium. We show that, without secondary stabilization mechanisms, price window designs cannot achieve both short- and long-term stability unless they are backed by already-stable reserves. In particular, the mechanism faces a tradeoff: either risk eventual reserve depletion through persistent arbitrage by a speculator, or widen the distance between mint and redeem prices enough to disincentivize arbitrage. In the latter case, however, the market price of the stablecoin inherits the volatility of its backing asset, with fluctuations that can be proportional to the backing asset’s own volatility.
%The abstract should briefly summarize the contents of the paper in 150--250 words. \kat{Don't forget this!}

% \keywords{stablecoins  \and price windows \and theoretical framework.}
\end{abstract}

\section{Introduction}

A stablecoin is a cryptocurrency whose price is ``pegged'' to some reference asset (typically the U.S. Dollar). Since their emergence in the mid-2010s, a wide range of designs and stabilization mechanisms have been developed. But despite strong incentives and many design attempts, no stablecoin has yet been shown to be both provably stable and fully decentralized. Designs backed by USD, for instance, have been shown to be stable when reserves are well-managed; but the management of off-chain assets requires trust -- the very thing cryptocurrencies were created to circumvent. By contrast, cryptocurrency-backed designs -- which can be managed by smart contracts instead of human custodians -- can avoid this issue, but have struggled with market price volatility. Algorithmic stablecoins seek to avoid a dependence on any external currency, instead controlling market price through regulation of supply and demand, much like a central bank. However, these designs have experienced some of the most catastrophic stablecoin failures to date (e.g., \vcite{p}{lms23}).

Two foundational questions in the field are therefore still open: Can a decentralized, stable coin even exist? And what is the full space of possible stablecoin designs? While our work in this paper represents only a very small part of the complete stablecoin picture, we believe that our approach -- in which we use simple mathematical models and minimal assumptions to distinguish between different families of stablecoin designs by the emergent properties of those systems -- could be built upon in future work. Our definitions and approaches were crafted with this in mind. 

\subsection{Our Contributions}
The scope of this paper is limited to what we call the ``price window model,'' in which the stablecoin mechanism mints and redeems stablecoins at fixed dollar prices, thereby constraining the market price equilibrium to lie between these bounds. This captures familiar USD-backed designs like Tether and USDC, as well as some crypto-backed designs like USDN and Cardano's Djed. Interestingly, these designs are by far the most common among stablecoins backed by USD, but seem to be less preferred among crypto-backed designs. Our work offers a possible explanation as to why.

In \vref{Section}{sec:model}, we suggest that what mathematically distinguishes these designs from other backed stablecoins is the information the mechanism can use to determine the mint and redeem prices it offers. While other designs (such as DAI) can use information about the particular transaction (e.g., the identity and history of the agent initiating the request), price window mechanisms treat all agents and stablecoins \textit{interchangeably}. This property leaves the mechanism vulnerable to a speculator who exploits fluctuations in the price of the backing asset (which we refer to throughout the paper as the ``backing coin'') to accumulate reserves.\footnote{By contrast, mechanisms like DAI with \textit{personalized} redeem prices can guard against this avenue of failure by blocking agents from redeeming stablecoins for more backing coins than they paid to mint -- but at the loss of the short-term stability afforded by price windows.} Informally, our main result in \vref{Section}{sec:general-results} shows that if the stablecoin mechanism is not allowed to use information about the specific transaction, then it cannot effectively guard against such attacks.

\begin{inftheorem}
    In the absence of secondary stabilization mechanisms, the long-term sustainability of a price window-based mechanism is only guaranteed when the size of the price window (that is, the difference between the mechanism's mint and redeem prices and, therefore, the rational price interval in the market) scales with the asymptotic volatility of the underlying asset.
\end{inftheorem}

This suggests a fundamental tradeoff between short-term and long-term stability in price window designs over unstable assets. Moreover, our result shows that the higher the volatility of the asset used to back the stablecoin, the higher the potential volatility of the stablecoin. This implies that -- at least, for designs which adopt a price window as the primary stabilization mechanism -- price stability seems to be largely an inherited property.

We expect \vreftwo{Sections}{sec:model}{sec:general-results} to be most relevant to theorists and researchers looking to define relevant stablecoin design taxonomies and understand emergent properties of different stablecoin designs; but it also has some practical relevance to system designers. For instance, our result provides a theoretical value for the size of transaction fees needed to prevent long-term depletion of a system's reserves. While this value might be overlarge for systems based on a highly volatile backing coin, it might be more reasonable to implement for mechanisms whose backing coin is close to stable.

\vref{Section}{sec:case-study} of this work is designed to answer some of the more practical questions raised by the first part of this work. For instance, in \vref{Section}{sec:general-results} we outline a theoretical mechanism by which some abstract speculator could deplete a system's reserves of backing coins, but it's unclear how seriously this risk should be taken. Our claim doesn't address any real-world concerns about whether a real speculator could successfully launch such an attack, and how powerful and well-resourced the speculator would need to be to deplete a system's reserves in a reasonable amount of time. To this end, we move from a theoretical black-box speculator to an explicit model of a speculator and compute a formula for the expected time it would take for the speculator to deplete a system's reserves, given that the backing coin prices are drawn independently at random from some known distribution. We explain how this formula relates the expected depletion time to the volatility of the backing coin, and how the same pattern we noticed in the first section -- that the more volatile the backing coin, the less sustainable the system -- becomes even more explicit here. We then relax the assumption of independent price draws and simulate the behavior of a speculator over historical Ether and Bitcoin price data, showing that a sufficiently wealthy speculator would have been able to drain the reserves of a stablecoin system backed by each of these currencies within three years.

\subsection{Related Work}

Understanding the source of stability in stablecoin design has been the focus of a growing literature since 2019, e.g., see \vcite{p}{mita2019stablecoin,lipton202011,lyons2023keeps,grobys2021stability,clements2021built,klages2020stablecoins}. The field of stablecoin literature is inherently interdisciplinary, with sources ranging from empirical and financial analyses \vcite{p}{du23,jk21} to theory \vcite{p}{cdkl+21,li2022money} to regulatory and legal reviews \vcite{p}{clements2021built,az21}. While each perspective offers an important layer to the stablecoin narrative, our approach in this paper is theoretical, and so our contribution is best compared against other theoretical papers. In general, the heterogeneity of stablecoins makes it difficult to propose theoretical claims which apply broadly across the design space, and many works respond to this challenge by either focusing on a single type of stablecoin design, or by proposing broader frameworks and taxonomies by which to better understand the options.

For these broader/taxonomic papers, on one end of the spectrum are papers like \vcite{p}{bkp19,moin2020sok,cd21,cds21}, which offer meaningful qualitative and intuitive observations but can lack mathematical formalism. The work of \vcite{t}{bkp19} is a paper of this type which has been of particular inspiration to our work, especially the authors' distinction between primary and secondary stabilization mechanisms.
%; although unlike us, they distinguish different backed stablecoins by whether their backing coin is on- or off-chain, while we distinguish according to the mechanism's protocols. 
Papers which take a more mathematical approach include \vcite{p}{klages2020stablecoins}, which classifies backed stablecoins according to the endo/exogeneity of the backing coins; \vcite{p}{zly21}, a fascinating paper which uses automata theory to develop a modeling and verification framework which generalizes several types of algorithmic stablecoins; and a 2024 paper by \vcite{t}{potter2024drives} which, of the taxonomic papers, seems closest to our own. In \vcite{p}{potter2024drives}, the authors generalize many known stablecoin designs by their redemption functions and then classify families of designs according to the economic conditions under which they can maintain the \$1 peg. However, unlike us, the authors assume that all value is measured in dollars, rather than backing coins. The only failure path for price-window based designs, according to this paper, is through undercollateralization if the price of the backing coin drops sufficiently -- something which can only occur in ``bad'' economic states. Meanwhile, our paper demonstrates how undercollateralization and reserve depletion can occur even in good economic states through a speculator who values backing coins more than dollars. This is a reasonable assumption to make for strong backing coins with high average returns, like Bitcoin or Ether.

With respect to limited-scope theoretical papers, our work differs in its focus on price window-based models over potentially volatile reserves. Since these designs are less commonly seen in practice, it is perhaps understandable that fewer papers have studied them as opposed to, for instance, fiat-backed stablecoins like Tether \vcite{p}{lyons2023keeps} or overcollateralized designs like DAI \vcite{p}{klages2021stability,klages2022while,b19,kv21}. Some notable papers include \vcite{p}{klages2021stability} and \vcite{p}{klages2022while}, which explore deleveraging spirals in stablecoins with DAI-like designs under the influence of an opportunistic speculator. \vcite{p}{b19} explores undercollateralization risks in BitShares, another overcollateralized, crypto-backed stablecoin. Although this paper explores a different setting from our work, the author's focus on understanding how the mechanism should choose its strategy to prevent some undesired behavior is similar to our own approach. \vcite{p}{dmv22} is another notable theoretical paper, but its setup explores systems in which redemption is not guaranteed -- instead, investors must sell their stablecoins through the competitive market. The platform defends the peg instead by contracting and expanding the supply of stablecoins. 

Several papers, like our own, note a link between the volatility of the backing coin and that of the stablecoin. For instance, \vcite{t}{kv21} show how unstable backing assets can distort a DAI-like system's mechanism for arbitrage -- in particular, the stablecoin tends to trade above (below) its peg when ETH is in a bad (good) state. \vcite{t}{potter2024drives} note the same behavior as well, though their motivation is wholly different from our own. While we argue that this ``inherited (in)stability'' occurs through arbitrage incentives, \vcite{t}{potter2024drives} argue that it instead arises from backing coin price discrepancies between agents' expectations and prices offered by the mechanism (due, e.g., to price oracle delays). The more volatile the backing coin, they argue, the more likely it is that these two values will be out-of-sync. In contrast to these papers, in our work we not only show the relationship between the volatility of the stablecoin and backing coin, but we also compute the size of the transaction fees needed to prevent eventual failure.

\section{The Model}\label{sec:model}
Let $T=\{1,2,\ldots\} \cong \mathbb{N}$ be a discrete and infinite sequence of \emph{timesteps}, which we will refer to as \emph{time}. In our setup, there are two types of coins: the \emph{tokens} and the \emph{backing coins}. Looking ahead, the former model the \emph{stablecoins}, and the latter model the (possibly volatile) assets that back the stablecoin system, e.g., Ether, BTC, or USD. The value (or \emph{price}) of both tokens and backing coins will be measured in terms of some \emph{fiat currency}; without loss of generality, we assume that this will be the U.S. dollar (\$). We will use $p_t$ to denote the dollar price of the backing coin at timestep $t \in T$ (provided, e.g., by an oracle), and assume that $p_t \geq 0$ for all $t \in T$ and that $p_t$ is exogenous to the system.

We define a \textit{token-generating mechanism} $\mathcal{M}$ with which some set of \textit{agents} interact by exchanging backing coins and tokens. The rules of this interaction are governed by the functions $\alpha$ and $\beta$, which capture, respectively, how the mechanism \textit{mints} and \textit{redeems} tokens, and which we describe below. Agents are also permitted to trade tokens and backing coins with \textit{each other} in a market. We let $d_t$ denote the equilibrium dollar market price of a token at time $t$. 

In this work (and in the theory of stablecoins more generally), we are interested in token-generating mechanisms which choose interaction functions $\alpha$ and $\beta$ with the aim of creating the conditions for a \textit{stable} token with market price $d_t = 1$.\footnote{This could be any positive constant; but according to convention and without loss of generality, we set it to 1.} We call such mechanisms ``stablecoin systems'' and the tokens they generate ``stablecoins,'' even if they are not always stable in practice. To clarify what we mean by stability, we introduce in this paper a notion which we refer to as ``weak $\varepsilon$-stability'':

\begin{definition}[$\varepsilon$-Stability over a sequence]
    Let $p_t$ be the price of the backing coin at time $t$. Given some $\varepsilon \geq 0$, we say that a stablecoin mechanism $\mathcal{M}$ is \textbf{$\varepsilon$-stable} over a sequence $(\frac{1}{p_t})_{t \in T}$ if $d_t \in [1-\varepsilon, \, 1+\varepsilon]$ for all $t \in T$. %\aris{do we need the set of agents here? The stability property is on the price, so where do the agents come into play?}
\end{definition}
\begin{definition}[Weak $\varepsilon$-stability]
    A mechanism $\mathcal{M}$ is \textbf{weakly $\varepsilon$-stable} for $\varepsilon \geq 0$ if there exists some non-convergent sequence $(\frac{1}{p_t})_{t \in T}$ over which $\mathcal{M}$ is $\varepsilon$-stable.
\end{definition}

We can conceptualize $(\frac{1}{p_t})_{t \in T}$ as the sequence of prices (in units of backing coin) over time for a stablecoin sold at $\mathdollar 1$. This notion of ``weak stability'' has two nice properties. For one, it seems clear that any mechanism which does \textit{not} satisfy this definition must be a fairly ineffective design. For if a design can only achieve stability over a sequence of eventually-convergent (in other words, eventually stable) prices, then the design does not offer any particular long-term improvement; that is, it does not expand the set of price sequences over which stability is achievable. The second nice property of the definition is that it makes a claim about the long-term \textit{sustainability} of a mechanism. It is not sufficient for the mechanism to provide some short-term protection against volatility; to be considered ``weakly stable,'' the mechanism must be able to offer this protection indefinitely over at least one sequence of non-convergent prices. (Note that a design can fail in catastrophic economic conditions (e.g., $p_t = 0$) and still be weakly $\varepsilon$-stable. This is important because otherwise, no backed stablecoin would ever be considered stable.)

\subsection{Design Taxonomy and Price Window}
Now that we've introduced the primary goal of any stablecoin mechanism, we turn to understanding the various ways these mechanisms try to achieve it. Stablecoin designs are often categorized by the primary and secondary stabilization mechanisms they use.\footnote{See \vcite{p}{bkp19} for a list of some primary and secondary stabilization mechanisms.} The first major taxonomic distinction we (and many others) make between families of stablecoins designs is whether the stablecoins are \textit{backed} or \textit{unbacked} (also known as ``algorithmic''). In this paper, we choose to only focus on the former. For backed stablecoin mechanisms, the value of a stablecoin is derived from agents' ability to redeem their stablecoins to the mechanism for $\beta(\cdot) \approx \mathdollar 1$ of some currency or asset of value (e.g., USD, ETH, BTC etc.). The currency of choice is what we refer to in this paper as the ``backing coin.'' The mechanism stores its supply of backing coins in its \emph{reserves} $R$. We let $R_t$ denote the number of backing coins in $R$ at timestep $t \in T$ (with initial state $R_0 > 0$), noting that the value can change at each timestep in response to agents' actions.

We also assume that agents must pay some amount $\frac{\alpha(\cdot)}{p_t}$ of backing coin to the mechanism to purchase/mint a stablecoin at time $t$.\footnote{While this assumption is not \textit{strictly} necessary theoretically, the burden of any backed stablecoin design, as we shall see, is to maintain sufficient reserves to preserve redeemability. Hence, a mechanism which does not charge users to mint but pays users to redeem is essentially lost before it has even begun.} Mathematically, we characterize backed stablecoins by these two operations, minting and redeeming.
%\footnote{In theory and in practice, there could be other agent-system interactions (referred to as ``secondary'' stabilization mechanisms) which are excluded from this model. For instance, many stablecoin designs have a third token, known as ``governance'' or ``equity'' tokens, and the ways these can be traded between agents and the mechanism requires an additional set of protocols. This is discussed further in the conclusion.} 
The mechanism's chosen functions $\alpha$ and $\beta$, along with the current backing coin price $p_t$, indicate how the mechanism maps a request to mint or redeem, respectively, a single stablecoin to some quantity of backing coins received/given in exchange. The information the mechanism uses to decide $\alpha$ and $\beta$ for each request (i.e., variables upon which $\alpha$ and $\beta$ depend) gives the next branch in our taxonomy.

Let $\theta_t$ represent the state of the system at time $t$, and let $\tau$ represent the particular transaction request. (The state $\theta$, for instance, might include the price of the backing coin and the amount of reserves held in the system, while $\tau$ might represent information about the agent initiating the request, as well as information about the stablecoin in question, such as when, by whom, and at what price it was minted.) These two variables together capture all the information upon which the stablecoin mechanism can determine its mint/redeem price functions $\alpha$ and $\beta$.\footnote{We might later wish to consider some random variable $X$, to allow the possibility that $\alpha$ and/or $\beta$ are stochastic functions; but for now we assume our stablecoin mechanism is deterministic.} From here, we make the following distinction: Any mechanism in which $\alpha$ and $\beta$ depend solely on the state of the system $\theta$ -- i.e., 
$$\alpha(\theta_t, \tau_i) = \alpha(\theta_t, \tau_j)\text{ and } \beta(\theta_t, \tau_i) = \beta(\theta_t, \tau_j) \text{ for all } \tau_i, \tau_j \text{ and for any } \theta_t$$
  -- is said to have agent and stablecoin \textit{interchangeability}, since the same mint and redeem prices at any time $t$ are offered to \textit{all} agents and stablecoins, regardless of their identity or history. 

These designs include ``tokenized'' stablecoins like USDC and USDT, as well as some crypto-backed designs like Cardano's Djed and USDN. (Note: The only difference between these two classes of designs, according to our model, is that the backing coins of the former are stable relative to the fiat currency. That is, for tokenized stablecoins, $\frac{1}{p_t} = 1$ for all $t \in T$.) By contrast, backed designs which use information about the specific agent and transaction to determine the redemption price of a stablecoin (like DAI), are considered \textit{non-interchangeable} under this taxonomy. We refer to the stabilization mechanism of interchangeable backed stablecoin designs as a \emph{\textbf{``price window''}} (for reasons we will explain shortly), and say that mechanisms which use this approach are following the ``price window model'' for stablecoin designs. 

To understand how such mechanisms work, suppose at time $t$ the mechanism offers some dollar mint/redeem prices $\beta(\theta_t) \leq 1 \leq \alpha(\theta_t)$ to all agents. Agents' ability to redeem stablecoins to the mechanism in exchange for $\mathdollar \beta(\theta_t)$-worth of backing coin creates a \textit{price floor} in the market (assuming no market frictions) -- for why would any rational agent sell their stablecoin to the market for less than $\mathdollar \beta(\theta_t)$, when they could obtain a higher value directly from the mechanism? Similarly, the ability to mint for $\mathdollar \alpha(\theta_t)$ creates a \textit{price ceiling} on the market. Together, the price floor and ceiling create a $\mathdollar[\beta(\theta_t), \alpha(\theta_t)]$ \textit{price window} around the equilibrium market price $d_t$ of a stablecoin, so long as the mechanism maintains its capacity to mint and redeem stablecoins at these prices. 

\subsection{Our Scope}
In what follows, we limit the scope of our analysis to stablecoins which follow a price window design. For the sake of simplicity, we limit this scope even further to price window designs where $\alpha()$ and $\beta()$ are constant with time, since these are the most common designs in practice and can be explored using simpler notation and proofs. We can think of these designs as charging fixed ``transaction fees'' $\varepsilon_\alpha, \, \varepsilon_\beta \geq 0$ for every mint or redeem request. However, our results -- especially in \vref{Section}{sec:general-results} -- are generalizable to transaction fees $\varepsilon_\alpha$ and $\varepsilon_\beta$ which are allowed to vary with time, although we do not include the proof in this paper.\footnote{Essentially, when $\varepsilon_\alpha$ and $\varepsilon_\beta$ are allowed to vary with time, \vref{Theorem}{thm:price-window-impossibility-result_constant-epsilon} still applies, but to $\lim \sup \varepsilon_\alpha(t)$ and $\lim \sup \varepsilon_\beta(t)$ rather than $\varepsilon_\alpha$ and $\varepsilon_\beta$.} We formally define the stablecoin mechanism $\mathcal{M}$ studied in this paper by the following minting and redeeming functions:
\begin{itemize}
    \item \textbf{Minting:} We define a minting price function $\alpha : \mathbb{N} \to \mathbb{R}_{\geq 0}$ which assigns to each time $t \in T \cong \mathbb{N}$ a \textit{mint price}, with units in USD. We write 
    $\alpha(t) = 1 + \varepsilon_\alpha$
    for some constant $\varepsilon_\alpha \geq 0$.
    i.e., $\MM$ will always mint a stablecoin for $(1+\varepsilon_\alpha)$-dollar's-worth of backing coins. In units of backing coin, this corresponds to a price of $\frac{1 + \varepsilon_\alpha}{p_t}$ for each time $t$. \smallskip 
    
    \item \textbf{Redeeming:} Similarly, for redeeming we define $\beta : \mathbb{N}\to \mathbb{R}_{\geq 0}$ as follows: for some constant $\varepsilon_\beta \geq 0$, $\beta(t) = 1 - \varepsilon_\beta$ if $R_t > \frac{1 - \varepsilon_\beta}{p_t}$ and $\beta(t) = R_t p_t$, otherwise.
    % \begin{align*}
    %     \beta(t) = \begin{cases}
    %         1 - \varepsilon_\beta & R_t > \frac{1 - \varepsilon_\beta}{p_t}, \\
    %         R_t p_t, & \text{otherwise,}
    %     \end{cases}
    % \end{align*}
    In other words, $\MM$ maintains a $(1 - \varepsilon_\beta)$-dollar stablecoin sale price whenever it has sufficient reserves; otherwise $\MM$ returns as many funds as it can, i.e., the entirety of its reserves. We say that the reserves of $\MM$ are eventually \textit{depleted} if there exists some $t \in T$ for which $R_t = 0$.
\end{itemize}

In this way, and assuming no market frictions, $\MM$ maintains a $\mathdollar [1-\varepsilon_\beta, 1 + \varepsilon_\alpha]$ price window around the market equilibrium price $d_t$ whenever it has sufficient reserves. The first timestep $t$ that $\MM$ is unable to maintain the $\mathdollar[1-\varepsilon_\beta, 1 + \varepsilon_\alpha]$ peg is the time of reserve depletion.

\section{General Results}\label{sec:general-results}
Now that we've introduced our model, we explain the failure path to which price-window stablecoin designs are susceptible and show that for small $\varepsilon$, these designs are \textit{not} weakly $\varepsilon$-stable. In fact, reserve depletion will eventually always occur under the sustained effort of a so-called ``sensitive speculator'' \textit{unless} the sequence of backing coin prices converges, or the size of the transaction fees is proportional to the volatility of the backing coin.

\subsection{Failure Pattern Intuition} Consider the lifecycle of a single stablecoin. Minting at backing coin price $p_{m}$ and redeeming at $p_{n}$ through a mechanism with a  $\mathdollar [1 - \varepsilon_\beta, 1 + \varepsilon_\alpha]$ price window changes the system's reserves by $\tfrac{1+\varepsilon_\alpha}{p_m} - \tfrac{1-\varepsilon_\beta}{p_n}$ units of backing coin. If $\tfrac{1+\varepsilon_\alpha}{p_m} > \tfrac{1-\varepsilon_\beta}{p_n}$, this impact is positive for the system; if not, it is negative. Over many coins, such effects average out if agents mint and redeem non-strategically. But if a speculator deliberately times minting and redemption to maximize these differences, the system can experience systematic reserve depletion, assuming both that the price of the backing coin is sufficiently volatile relative to the sizes of $\varepsilon_\alpha, \varepsilon_\beta$, and that the speculator is sufficiently skilled at recognizing and acting on backing coin price anomalies. And since any backing coin the speculator gains must be taken from the mechanism's reserves, eventually, sustained arbitrage will exhaust those reserves.%\footnote{\label{footnote:no-market} The reader may have noticed that this description suggests the speculator interacts exclusively with $\mathcal{M}$ an not with the market. \kat{Explain.}}

% To capture the notion of such a speculator without sidetracking our global analysis with a mathematically unwieldy object, we choose to model the speculator as essentially a ``black-box'' in this section. (Section \todo{ref}, by contrast, explores what happens when we instantiate a specific speculator.) This black-box speculator is defined below.

\subsection{Defining a ``Sensitive'' Speculator} 
Let $\mathcal{M}$ denote a backed stablecoin mechanism with mint and redeem functions $\alpha(t)$ and $\beta(t)$. Define $s_t$ as the total profit (in units of the backing coin) that could be attained by an optimal, omniscient speculator who: 1) begins with a single unit of backing coin, 2) interacts exclusively with $\mathcal{M}$, and 3) knows the entire sequence of backing coin prices $(\frac{1}{p_t})_{t \in T}$ in advance. Then $s_t$ represents the maximum profit achievable up to timestep $t$ over this fixed sequence of prices.

\begin{definition}[Approximate optimality]
    Suppose a speculator has some number of backing coins $n_0$ at $t=0$. We say that a speculator's strategy is $\boldsymbol{\frac{1}{k}}$\textbf{-approximately optimal} for some $k \geq 1$ if the total profit $r_t$ obtained by the speculator up to any time $t$ is at least $\frac{1}{k}$th that of the optimal value; i.e., $r_t \geq \frac{n_0}{k} s_t$ for all $t \in T$. 
\end{definition}

This definition essentially says that if the mechanism's transaction fees and the sequence of prices offer an opportunity to profit, then the speculator can generally recognize and act on this opportunity, albeit imperfectly. In this way, the definition can be thought of as an assumption about the speculator's ``sensitivity'' to price fluctuations. We formalize this idea in our next definition.
\begin{definition}[Speculator sensitivity]
    A speculator is \textbf{sensitive} if there exists some number $k \geq 1$ such that the speculator's strategy is $\frac{1}{k}$-approximately optimal over any sequence of prices.
\end{definition}
Conveniently, these definitions allow us to construct general results about the asymptotic global dynamics of price window-based stablecoins while making very few assumptions on the internal details of the speculator's strategy or behavior. This leads us to our first main result, the proof of which can be found in \appref{\vref{Appendix}{app:limsup-inf-proof}}{the full-text version}.

\begin{theorem}\label{thm:price-window-impossibility-result_constant-epsilon}
    Suppose a backed stablecoin mechanism $\mathcal{M}$ has finite reserves and constant mint and redeem functions $\alpha(t) = 1 + \varepsilon_\alpha, \, \beta(t) = 1 - \varepsilon_\beta$ for $\varepsilon_\alpha, \varepsilon_\beta \geq 0$. Define $L := (1 + \varepsilon_\alpha) \lim \inf \tfrac{1}{p_t} - (1 - \varepsilon_\beta) \lim\sup \tfrac{1}{p_t}$ and $\varepsilon := \max\{\varepsilon_\alpha, \varepsilon_\beta\}$. Then $\MM$ is $\varepsilon$-stable over $(\frac{1}{p_t})_t$ in the presence of any sensitive speculator if $L > 0$, and only if $L \geq 0$.\footnote{But in practice, \vref{Theorem}{thm:price-window-impossibility-result_constant-epsilon} is almost always an equivalence. We explain why in \appref{\vref{Appendix}{app:limsup-inf-proof}}{the full-text version}.}
\end{theorem}
    
\subsection{Interpretation}
Recall that $\lim\sup \tfrac{1}{p_t} := \lim_{N \to \infty} \sup\left\{\tfrac{1}{p_t} : t > N\right\}$ and that $\lim\inf \tfrac{1}{p_t} := \lim_{N \to \infty} \inf\left\{\tfrac{1}{p_t} : t > N\right\}$. To interpret \vref{Theorem}{thm:price-window-impossibility-result_constant-epsilon}, we begin with the simple case of $\varepsilon_\alpha = \varepsilon_\beta = 0$. In this case, our claim implies that reserve depletion occurs if and only if $\lim \inf \tfrac{1}{p_t} = \lim \sup \tfrac{1}{p_t}$,
%\kat{This should be a strict inequality.} 
which, by the squeeze theorem, implies the original sequence $(\tfrac{1}{p_t})_{t \in T}$ converges. Thus, no price window-based stablecoin mechanism which mints and redeems at \$1 can be weakly 0-stable. The only designs of this type capable of both short- and long-term stability, therefore, are those for which $(\frac{1}{p_t})_{t \in T}$ always converges -- for instance, USDT and USDC, where $\frac{1}{p_t} = 1$ for all $t \in T$.

In the more general case, we note that the total profit of sensitive speculators can be scaled to be \textit{arbitrarily small}. As a consequence, the only sequences over which any such speculator is guaranteed to deplete the system's reserves must have infinite total profit. Since each discrete time interval only offers finite profit (even for an optimal speculator), the only way profit can become infinite is in the tail in the price sequence. This explains why our criteria for $\varepsilon_\alpha$ and $\varepsilon_\beta$ (and the subsequent proofs) are only concerned with the behavior of the \textit{tail} of the price sequence.

To understand more specifically the role of $\lim\inf\frac{1}{p_t}$ and $\lim\sup\frac{1}{p_t}$, we can interpret these values as numbers which capture the range of values in the tail, since $\lim\sup$ is the tail's least upper bound and $\lim\inf$ is its greatest lower bound. If $\lim\inf\frac{1}{p_t}$ and $\lim\sup\frac{1}{p_t}$ are both finite, then the quantity $\lim\sup\frac{1}{p_t} - \lim\inf\frac{1}{p_t}$ represents the range or spread of the tail of the price sequence. This value can be thought of as a measure of the price sequence's asymptotic \emph{volatility}: the larger the difference between $\lim\inf\frac{1}{p_t}$ and $\lim\sup\frac{1}{p_t}$, the more volatile the asymptotic behavior of the price sequence. 

In this way, charging ``transaction fees'' can be thought of as a way of flattening the spread of mint and redeem prices (scaling $\lim\sup\frac{1}{p_t}$ and $\lim\inf\frac{1}{p_t}$ by $(1 - \varepsilon_\beta)$ and $(1 + \varepsilon_\alpha)$, respectively) to decrease the range over which a speculator can asymptotically profit. Our claim essentially says that a mechanism can be weakly $\varepsilon$-stable if and only if $\varepsilon_\alpha$ and $\varepsilon_\beta$ are large enough to make this ``scaled'' spread asymptotically tend to zero.

If we set $\varepsilon_\alpha = \varepsilon_\beta = \varepsilon$, then, according to our result, $\varepsilon$ must satisfy 
\begin{align*}
    \varepsilon \geq \frac{\lim \sup \tfrac{1}{p_t}  - \lim\inf \tfrac{1}{p_t}}{\lim\sup \tfrac{1}{p_t} + \lim \inf \tfrac{1}{p_t}}
\end{align*} 
in order to prevent reserve depletion. Here, we can directly see the relationship between the threshold value of $\varepsilon$ and the price sequence's volatility: the two quantities are proportional to one another. (Note that the denominator, $\lim\sup \tfrac{1}{p_t} + \lim \inf \tfrac{1}{p_t}$, can be interpreted simply as a normalizing factor. It is unrelated to the volatility of the price sequence.) As a result, the range of the price window offered by the mechanism becomes proportional to the spread in the original price sequence $(\frac{1}{p_t})_{t \in T}$. Recall that the key advantage of price windows is their ability to constrain (rational) agents to only trade stablecoins on the market at prices within the offered price window. This is the strategy by which such mechanisms maintain stability. However, our result shows that in order to be sustainable, the size of this price window must essentially be as large as the volatility of the original price sequence. As a result, the market price of stablecoins is allowed to be, proportionately, as volatile as the backing coin against which the system aims to protect.

Of course, transaction fees do help delay a system's time to depletion by adding more reserves to the system as a buffer against arbitrage attacks. However, used in isolation, we conclude that price windows are insufficient to simultaneously ensure both short-term and long-term stability over a volatile backing coin.

\section{Case Study: A Speculator with IID Price Draws}\label{sec:case-study}
The theoretical advantage of the previous section lay in its generality and its ability to capture system-level dynamics with very few assumptions. However, almost necessarily, this generality comes at the cost of overlooking some practical questions that are useful to system designers and managers. In this section, we explore how such specifics can be modeled within the broader framework outlined in \vref{Section}{sec:general-results}. For instance, how likely is this failure to occur in practice? How powerful must a speculator be, and how long must they persist, to inflict real damage on the system? 

To address these questions, we forego our black-box treatment and implement a concrete model of a speculator. Since our earlier claims relied on the existence of a sensitive speculator, here we propose a possible implementation and compute the expected time required for reserve depletion under the assumptions of no transaction fees (i.e., $\varepsilon_\alpha = \varepsilon_\beta = 0$) and independent, identically distributed (i.i.d.) backing coin prices drawn from a known distribution. Note that this section is intended as a high-level summary rather than an in-depth technical analysis. For a full breakdown of the definitions, equations, and results referenced here, we direct the reader to \appref{\vreftwo{Appendices}{app:speculator}{app:expected-portfolio}}{the full-text version}, where each concept is laid out in detail.

\subsection{Modeling the Speculator}
In this section, we consider a speculator who has intrinsic value for backing coins but values stablecoins only as a means of obtaining more backing coins. The speculator interacts with a price window-based mechanism $\mathcal{M}$ with $\alpha(t) = \beta(t) = 1$ over a sequence of timesteps, and at each timestep, in response to the revealed backing coin price $p_t$, either chooses some number $\Delta_t$ of stablecoins to buy ($\Delta_t > 0$) or sell ($\Delta_t < 0$); or simply waits until the next timestep ($\Delta_t = 0$). 
The speculator has some expected utility function $u$ over portfolios of stablecoins and backing coins and chooses $\Delta_t$ as the number which maximizes their expected utility. 

In our model, we assume the speculator's utility function depends on three parameters: their level of (im)patience $\delta$, their beliefs about backing coin prices, and their risk tolerance $\Lambda$. Specifically, $\delta \in [0, 1)$ captures the rate at which the speculator discounts future value, with low $\delta$ indicating a patient speculator, and high $\delta$ an impatient one. For speculator beliefs about backing coin prices, we assume 1) that backing coin prices are drawn independently at random from some distribution with probability density function $f(x)$, and 2) that $f(x)$ is known to the speculator. Finally, we define the variable $\Lambda \in [0, 1]$ to indicate the speculator's risk tolerance, with $\Lambda = 1$ representing extreme risk aversion and $\Lambda=0$ representing risk neutrality.
%\footnote{As in Section \ref{sec:general-results}, we assume that there is only a single representative speculator, and that the speculator chooses only to interact with the mechanism $\mathcal{M}$, rather than buying or selling stablecoins through a secondary market.\kat{Explain these assumptions / reference remarks in earlier section.}} 
To simplify our analytic results, we also assume that the mechanism $\mathcal{M}$ uses a price window design with $\varepsilon = 0$. %, although we relax this assumption in our simulations. \kat{Do this or change sentence.} 

Under this model, the speculator’s strategy can be seen as a tradeoff: on one hand, they might use all their funds to capitalize on any opportunity for profit, however small, but risk being unable to act when truly exceptional prices occur. On the other hand, they can wait for rare opportunities with exceptional prices, accepting delayed gratification in hopes of obtaining larger profits later.

\subsection{Understanding the Conditions for (In)Stability}
\label{sec:unpacking-stability}
In \appref{\vreftwo{Appendices}{app:speculator}{app:expected-portfolio}}{the full-text version}, we show how we can solve the speculator's maximization problem at each timestep $t$ and, assuming prices are drawn independently at random from our known distribution, compute a general formula for the expected number $n_t$ of backing coins the speculator owns at any time $t$ for $n_0 > 0$. From here, we can find the expected time to depletion of a mechanism with an initial number $R_0$ of backing coins in its reserves, though it is not always possible to solve for this time using analytic methods unless $\Lambda = 0$.

\begin{example}\label{ex:lambda=0}
    Suppose backing coin prices are drawn uniformly at random from a normal distribution with $\mu = \sigma^2 = 100$, and suppose the speculator chooses $\Lambda=0$ and $\delta = 0.1$. Then setting the ratio $\frac{R_0}{n_0}$ to 100 and computing $a_1$ according to the formulas in \appref{\vref{Lemma}{Lemma:eigenvalues} in \vref{Appendix}{app:expected-portfolio}}{the full-text version} gives a reserve depletion which is expected to occur at approximately $k(t) = \log(1 + \frac{R_0}{n_0}) / \log(a_1) \approx 15.78$, which in turn corresponds to approximately $227$ timesteps.\footnote{See \appref{\vref{Appendix}{app:k-to-t}}{the full-text version} on how to convert from $k(t)$ to the number of expected timesteps.} We remark that a ``timestep'' is designated by each time the speculator interacts with the stablecoin mechanism. Assuming, for example, relatively frequent interactions of once every hour, as well as the price updates instructed by our distributional assumptions, this calculation shows that depletion of resources will take place in less than 10 days. 
\end{example}

 Our formula for the speculator's expected portfolio also allows us to show the analogous ``convergence result'' to \vref{Theorem}{thm:price-window-impossibility-result_constant-epsilon}. In \appref{\vref{Appendix}{app:speculator-reserve-depletion-results}}{the full-text version}, we show that the speculator's optimal strategy is to buy stablecoins whenever the backing coin price $p_t$ falls above some threshold price $y_2$, and to sell whenever $p_t$ falls below some price $y_1 \leq y_2$. From here, we prove that, assuming independent price draws from a distribution with probability density function $f(x)$ and cumulative distribution function $F(x)$, a speculator with $n_0 > 0$ can never deplete a system's reserves if and only if $Y = \frac{\frac{1}{1 - F(y_2)} \int_{y_2}^\infty x \, f(x) \, dx}{\frac{1}{F(y_1)} \int_{-\infty}^{y_1} x \, f(x) \, dx} = 1$ (or if the speculator is so risk averse that $\Delta_t = 0$ for all $t \in T$). Moreover, the larger the value of $Y$, the more rapidly the reserves are depleted. We can use \appref{\vref{Lemma}{lemma:normal_xf(x)_expansion} in \vref{Appendix}{app:normal-distr-expansion-of-Y}}{a formula from the full-text version} to expand the functions $f(x)$ and $F(x)$ in the case of a normal distribution as
 $$Y = \frac{ \mu + \sigma^2 \frac{f(y_2)}{1 - F(y_2)}}{\mu - \sigma^2 \frac{f(y_1)}{F(y_1)}}.$$

From this representation, the relationship between a distribution's volatility and $Y$ is clear: $Y = 1$ if and only if $\sigma^2 = 0$, and the larger the volatility $\sigma^2$, the larger the value of $Y$. This result corresponds nicely to the more general result in \vref{Theorem}{thm:price-window-impossibility-result_constant-epsilon}, where we showed the relationship between the range of a price sequence's tail and the opportunity for a speculator to deplete the system's reserves. Here, we see that a large value for $Y$ corresponds to a large interval surrounding the tail of the price sequence.

We can see the same relationship in our numerical simulations. Consider, for instance, \vref{Figure}{fig:different-deltas}, where we note that the expected time to reserve depletion increases \textit{dramatically} for small values of $\sigma^2$ -- note that the scale of the $y$-axis is logarithmic. %\kat{If we include $\varepsilon$, comment on this, too.} 
In the next section, we show that the same behavior arises even when we relax the assumption of independent price draws.

\begin{figure}[h!]
    \centering
    \includegraphics[width=0.65\linewidth]{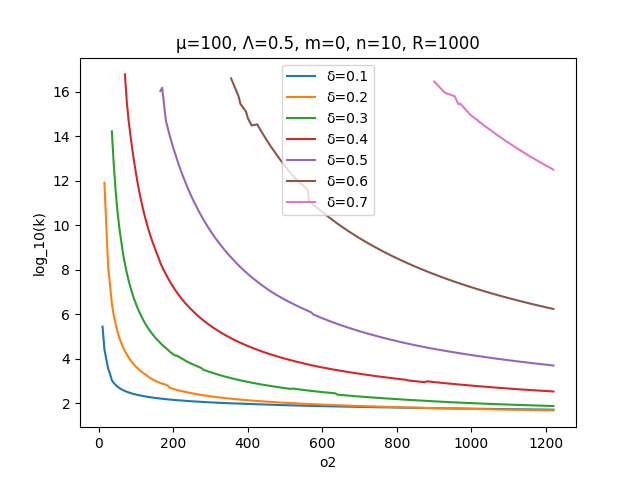}
    \caption{The log of the expected time to depletion for prices drawn independently at random from a normal distribution with $\mu = 100$ and various values of $\sigma^2$, given a speculator with $\delta \in [0.1, 0.2, 0.3, 0.4, 0.5, 0.6, 0.7]$, $\Lambda = 0.5$, $n_0 = 10$, and initial reserves $R_0 = 1000$.}
    \label{fig:different-deltas}
\end{figure}

\subsection{Simulations}

    For our results in \vref{Section}{sec:unpacking-stability}, we assumed that 1) the price of the backing coin in each timestep is drawn independently at random from some distribution; and 2) the speculator believes that this is the case. While this assumption of independence allows us to produce exact, analytic solutions, in many cases it might not always accurately reflect price fluctuations in reality. In this section, we demonstrate through simulations that the risk of reserve depletion does not depend on the assumption of price independence. Moreover, we find that, as in the case of independent prices, the time it takes for reserves to be depleted decreases as the volatility of the prices grows. Therefore, our simulations convey essentially the same message, just in an environment where the price updates are not independent. This is unsurprising, given the general result shown in \vref{Theorem}{thm:price-window-impossibility-result_constant-epsilon}.
    
    To achieve this, we coded an adapted version of our speculator and designed a simulator which, given a sequence of backing coin prices\footnote{This could be a fixed dataset; or could be constructed in real-time e.g., by repeatedly querying some random process. We use both approaches in different simulations.} and an initial reserve value $R_0$, repeatedly updates the reserves $R_t$ according to the speculator's response to the revealed backing coin price $p_t$. The simulation stops either when the reserves have been depleted ($R \leq 0$), or the pre-fixed maximum number of iterations (in our case, 100,000) has been reached.  
    Note that, unlike in the case of independent price draws, the speculator’s strategy no longer guarantees that every stablecoin is sold at a higher price than it was bought, nor that reserves are eventually depleted when $\sigma^2 > 0$.

\subsection{Simulating with Real-World Data}
To demonstrate that, without any interventions or contingency mechanisms, reserve depletion is nevertheless a very real possibility in real-world stablecoin systems, we collected historical price data on Bitcoin (BTC) and Ethereum (ETH) -- the two most popular cryptocurrencies in the present day -- from April 2022 to April 2025. Each data point represents the hourly opening price of that currency on the online trading platform Coinbase \vcite{p}{coinbase}, giving us a sequence of roughly 26,000 prices for each currency. We then simulated the speculator's response against each sequence to see whether reserve depletion could occur during this three-year time frame. We assumed that each of the reserves had an initial store of $R_0 = 1000$ coins, and we calculated the smallest value $R_{min}$ the reserves hit during the course of the simulation, given different initial values for $n_0$. The results are shown in \vref{Table}{table:ETH-vs-BTC-R_min-simulations}, where each entry in the table represents $R_{min}$ for a particular currency, given an initial speculator endowment of $n_0$ coins.

% Requires: \usepackage{booktabs}
\begin{table}[t]
  \centering
  \begin{tabular}{|c|c c|}
    \hline
    Initial endowment of $\mathbf{n_0}$ & \textbf{ETH} & \textbf{BTC} \\ \hline
    100 & \; 854.82 \; & \; 879.59 \; \\ 
    250 & \; 637.04 \; & \; 698.96 \; \\ 
    500 & \; 274.08 \; & \; 397.93 \; \\ 
    750 & 0 & 96.89 \\ 
    1000 & 0 & 0 \\
    \hline
  \end{tabular}
   \vspace{0.3cm}
  \caption{The smallest value $R_{min}$ of the reserves during each simulation using three years' worth of historical price data from each Ethereum (ETH) and Bitcoin (BTC), for various speculator initial endowments $n_0$ and initial reserves $R_0 = 1000$.}
  \label{table:ETH-vs-BTC-R_min-simulations}
\end{table}

Thus, we see that given a large enough initial endowment, reserve depletion eventually occurred for both types of backing coins -- even when the speculator used a strategy that had not been optimized specifically for real-world prices. It is also notable that reserves of Ether depleted more quickly in practice than reserves of Bitcoin. This perhaps can be attributed to the fact that Ether historically has had higher rates of volatility than Bitcoin. To explore this hypothesis further, we examine the impact of different levels of price volatility on both the probability of reserve depletion within a certain time frame, and the time to reserve depletion for different sets of simulated prices.

In particular, the price fluctuations of stocks in financial markets are often modeled as (Gaussian) random walks (e.g., see \vcite{p}{fama1970efficient,fama1965behavior,cootner1964random}, and \vcite{p}{urquhart2016inefficiency,aggarwal2019bitcoins} for a discussion specifically related to cryptocurrencies), in which the \emph{change} in prices (rather than the prices themselves) are drawn independently at random from some normal distribution with average step size $\mu_{\text{step}}$ and standard deviation $\sigma_{\text{step}}$. In our simulations, we compute the average value of $\mu_\text{step}$ and $\sigma_\text{step}$ over our three-year collection period for each BTC and ETH: The average values we computed from our datasets and used in our simulations were $\mu_\text{step} \approx -0.056, \, \sigma_\text{step} \approx 16.7, \, p_0 = \$ 3047.70$ for Ethereum, and $\mu_\text{step} \approx 1.66, \, \sigma_\text{step} \approx 302, \, p_0 = \$40508.04$ for Bitcoin. Then, we simulate random walks starting from the same price $p_0$ and with the same value of $\mu_\text{step}$, but with different values of $\sigma_\text{step}$, to show how both the probability of reserve depletion within 100,000 time steps, as well as the expected time-to-depletion of those simulations which \textit{do} end in reserve depletion, change with $\sigma_\text{step}$. For each value of $\sigma_\text{step}$, we run 100 different random walks and aggregate the results. The outcome of these simulations is shown in \appref{\vref{Figure}{fig:BTC-and-ETH_random_walks} of \vref{Appendix}{app:simulations}}{the full-text version}.

For both currencies, we can see a clear trend: as the average step size $\sigma_\text{step}$ between prices increases, so too does either the probability of reserve depletion (in the case of Bitcoin), as well as the rate of reserve depletion (for both currencies; but for Ether especially). When $\sigma_\text{step}$ is very close to zero, the reserves are practically never depleted. %\kat{What about including $\varepsilon$ here? Attempt if time.}

\section{Future Work}
Throughout the paper, we made several simplifying assumptions to either make our work tractable or to highlight what we believed to be the important features of the model without adding too much complexity. In this section, we bring attention to some of these assumptions, explaining why they were made and what impact removing them would likely have.

\paragraph{No competitive market.} As readers may have noticed, we do not directly model a competitive market in this paper; instead, we simply argue that the price window constrains the equilibrium market price of a stablecoin. This exclusion stems from two motivations: 
First, we do not believe adding a market would alter the system’s long-term behavior and emergent properties so long as market funds are finite: a speculator might arbitrage against the market for a time, effectively delaying reserve depletion, but once those funds are exhausted, the system's reserves face the same fate. Essentially, there is a conservation principle at play here: finite market funds and system reserves cannot withstand infinite arbitrage opportunities. The second reason we excluded a market from this model is that, though the intuitive effect is simple, we were unable to find an equally simple representation that captured the important information but did not over-complicate our model. In future work, we would like to add a minimal representation of a market, akin to our notion of a sensitive speculator.

\paragraph{A single speculator.} 
In our model, we assume there are no other agents besides a single speculator. This is done primarily for the sake of simplicity, and because the underlying dynamics of a single speculator on the system are similar even in the presence of other agents, if somewhat dampened and obfuscated beneath the other signals. In particular, multiple non-strategic speculators have an additive impact on reserve depletion, and simply decrease the expected time to reserve depletion. The presence of other non-speculative agents (who wish to buy/sell stablecoins only to use them, not to profit from arbitrage over price fluctuations) can extend the time to reserve depletion by adding additional funds to the stablecoin system's reserves; but if the agents are non-speculative, then there will not be a strong correlation between buy/sell times of these agents and the price of the backing coin. Thus, the underlying impact of the speculator is still present, if somewhat obscured by other non-speculative trades. However, modeling the interactions with multiple \textit{strategic} speculators would be an interesting extension of this work.

\paragraph{Secondary stabilization mechanisms.} This paper showed that a mechanism which depends on price windows as its only stabilization mechanism cannot be both short- and long-term stable. In our opinion, the most interesting extension of this work would be to understand how and whether secondary stabilization mechanisms, such as governance or equity tokens, would change this behavior. Many stablecoin designs include such ``dual'' tokens, the idea being to sell these tokens cheaply when reserves are low to raise funds, and to allow agents to redeem them at higher prices in the future. In the spirit of our stated goal of formalizing a mathematical framework, we would like to understand not just how known secondary mechanisms impact the probability and rate of reserve depletion, but also whether the known designs represent the entirety of the possible design space. For this, we suspect that a full or partial axiomatization of the stablecoin design space could help us restrict the scope of these secondary mechanisms.

\bibliographystyle{plainnat}
\bibliography{content/references}

\clearpage
\appendix

\section*{APPENDIX}

\section{Proof of weak \texorpdfstring{$\varepsilon$}{Ɛ}-stability conditions}\label{app:limsup-inf-proof}

\begin{proof}[Proof of Theorem \ref{thm:price-window-impossibility-result_constant-epsilon}]\label{pf:price-window-impossibility-result_constant-epsilon}
    We segment this proof into two claims, which we prove separately.
    \begin{enumerate}[(1)]
        \item Suppose our mechanism $\mathcal{M}$ has finite initial reserves $R_0 > 0$. Then any sensitive speculator with $n_0 > 0$ can eventually deplete the system's reserves if and only if the total profit attainable by an \textit{optimal} speculator tends to infinity (i.e., $\lim_{n \to \infty} s_n \to \infty$).
        \item $\lim_{n\to\infty} s_n = s \in \mathbb{R}$ if $L = (1 + \varepsilon_\alpha) \lim \inf \tfrac{1}{p_n} - (1 - \varepsilon_\beta) \lim\sup \tfrac{1}{p_n} > 0$, and only if $L \geq 0$.
    \end{enumerate}
    \begin{proof}[Proof of (1)]
    Suppose there exists a sensitive speculator; i.e., there exists some $k \geq 1$ such the total profit $r_n$ obtained by the speculator up to any time $n$ satisfies 
    $$r_n \geq \frac{n_0}{k} \, s_n.$$
    Then the total profit of the speculator as $n \to \infty$ is lower-bounded by $\lim_{n \to \infty} \frac{n_0}{k} \, s_n = \frac{n_0}{k} \lim_{n \to \infty} s_n$. Now suppose $\lim_{n \to \infty} s_n \to \infty$. Then by the definition of divergence to infinity, for any $M > 0$, there must be some $N \in \mathbb{N}$ for which $n > N$ implies $s_n > M$. Choosing $M = \frac{k R_0}{n_0}$ gives us
    $$r_n \geq \frac{n_0}{k} s_n > \frac{n_0}{k} M = R_0$$
    for any $n > N$. Therefore, any sensitive speculator (in the absence of other agents or protocols) can eventually deplete the system's reserves.
    
    On the other hand, if $\lim_{n \to \infty} s_n = s$ for some $s \in \mathbb{R}$, then we can construct a sensitive speculator who, given initial endowment $n_0$, would be unable to ever deplete the system's reserves, which we assume to have initial value $R_0 > 0$. 
    
    In particular, consider a speculator whose profit up to each timestep $n$ is exactly $r_n = \frac{n_0}{k} s_n$ for $k > \frac{s n_0}{R_0}$. (Note: Such a speculator can be shown to exist by construction: simply let the speculator buy/sell at the same timesteps as the optimal one, but with value weighted by $1/k$.) Such a speculator is sensitive by definition, and yet their total profit as $n\to\infty$ is strictly less than $R_0$, since
    $$\lim_{n \to \infty} r_n = \frac{n_0}{k} \lim_{n \to \infty}  s_n < \frac{n_0}{\tfrac{sn_0}{R_0}} s = R_0.$$
    Therefore, not every sensitive speculator is capable of forcing a reserve depletion when $\lim_{n \to \infty} s_n$ is finite. This concludes the proof for Claim 1.
    \end{proof}
    \begin{proof}[Proof of (2)]
        First, we assume that both $\lim\sup \frac{1}{p_n}$ and $\lim\inf \frac{1}{p_n}$ are finite. This is without loss of generality: for if $\lim\sup \frac{1}{p_n} = \infty$, then the original price sequence $(p_n)_n$ must tend to zero. And if $\lim_{n \to \infty} p_n = 0$, then the value of the reserves, $R_n p_n$, also tends to zero, regardless of any actions of the speculator. If $\lim\sup \frac{1}{p_n}$ is finite, then so too must be $\lim\inf \frac{1}{p_n}$, since $0 \leq \lim\inf\frac{1}{p_n} \leq \lim\sup\frac{1}{p_n}$.

        % IDEA: If the price window is large enough (i.e., the property holds), then after a certain point, the sup and inf become arbitrarily close to their limits, and so the max profit attainable between any two timesteps past that point is non-positive. 
        Now suppose $\varepsilon_\alpha, \, \varepsilon_\beta \geq 0$ satisfy $L = (1 + \varepsilon_\alpha) \lim\inf \tfrac{1}{p_n} - (1 - \varepsilon_\beta) \lim\sup \tfrac{1}{p_n} > 0$. By the limit definitions of $\lim \sup$ and $\lim \inf$ (where for any sequence $(x_n)_{n \in \mathbb{N}}$, we have $\lim\sup x_n := \lim_{N \to \infty} \sup\{x_n : n > N\}$ and $\lim\inf x_n := \lim_{N \to \infty} \inf \{x_n : n > N\}$), then for any $c > 0$, we can find some $N \in \mathbb{N}$ such that both of the following statements hold:
    \begin{itemize}[$\bullet$]
        \item $\sup \{ \frac{1 - \varepsilon_\beta}{p_n} : n > N\} - \lim \sup \frac{1 - \varepsilon_\beta}{p_n} < c$,
        \item $\lim \inf \frac{1 + \varepsilon_\alpha}{p_n} - \inf \{ \frac{1 + \varepsilon_\alpha}{p_n} : n > N\} < c$.
    \end{itemize}
    
    In particular, there exists some $N$ such that the statements above hold for $c = \frac{L}{2}$.
    We define $s_N^* := \sup\{ \frac{1 - \varepsilon_\beta}{p_n} - \frac{1 + \varepsilon_\alpha}{p_m} : m < n < N \}$ to be the largest difference between any two price-window-adjusted mint and redeem prices $\frac{1 + \varepsilon_\alpha}{p_m}$ and $\frac{1 - \varepsilon_\beta}{p_n}$, respectively, with $N < m < n$; it follows immediately that $s_N^* \leq \sup \{ \frac{1 - \varepsilon_\beta}{p_n} : n > N\} - \inf\{\frac{1 + \varepsilon_\alpha}{p_n} : n > N\}.$ From here, it follows that 
    \begin{align*}
        s_N^* &\leq \sup \left\{ \tfrac{1 - \varepsilon_\beta}{p_n} : n > N \right\} - \inf\left\{\tfrac{1 + \varepsilon_\alpha}{p_n} : n > N\right\} \\
        &< \left(\lim \sup \tfrac{1 - \varepsilon_\beta}{p_n} + \tfrac{L}{2}\right) - \left(\lim \inf \tfrac{1 + \varepsilon_\alpha}{p_n} - \tfrac{L}{2}\right) \\
        &= \left(\lim \sup \tfrac{1 - \varepsilon_\beta}{p_n} - \lim \inf \tfrac{1 + \varepsilon_\alpha}{p_n}\right) + L \\
        &= \left((1 - \varepsilon_\beta)\lim \sup \tfrac{1}{p_n} - (1 + \varepsilon_\alpha)\lim \inf \tfrac{1}{p_n}\right) + L \\
        &= 0.
    \end{align*}
    Therefore, the speculator is unable profit after timestep $N$, meaning their total profit must be finite.
    
    \paragraph{A note on equivalence.} In most cases, Claim 2 is an equivalence. The only degenerate boundary case is when $L = 0$ and the price sequence $(\frac{1}{p_n})_n$ is such that the amount of profit the speculator can obtain between any two timesteps tends to zero, but the \textit{sum} tends to infinity. An example of such a price sequence would be the harmonic sequence $1 - \frac{1}{2}, 2 + \frac{1}{2}, 1 - \frac{1}{3}, 2 + \frac{1}{3}, 1 - \frac{1}{4}, 2 + \frac{1}{4}, \ldots$. This sequence never converges, and setting $\varepsilon_\alpha = \varepsilon_\beta = \frac{1}{3}$ implies $L = 0$. The optimal strategy would be to buy/sell at each timestep, leading to a sequence of profits $(\frac{2}{n})_n$. Each term therefore tends to zero, but the sum of profit starting at any $N$ is $\sum_{n=N}^{\infty} \frac{2}{n} \to \infty$. Note that this pattern requires the sequence to take values in a continuous range. When the sequence is restricted to a discrete set of values (as with most price sequences), such a case cannot occur and so Claim 2 holds with equivalence.
    
    \paragraph{} Conversely, suppose $\lim_{n \to \infty} s_n = s \in \mathbb{R}$. Then the differences between the price window-adjusted mint and buy prices must tend to zero -- for otherwise, there would be infinite opportunity to profit as $n \to \infty$. Formally, we can write that for any $\delta > 0$, there exists $N \in \mathbb{N}$ such that for all $m, n > N$, we have 
    \begin{align*}
        (1 - \varepsilon_\beta) \tfrac{1}{p_n} - (1 + \varepsilon_\alpha) \tfrac{1}{p_m} < \tfrac{\delta}{2}.
    \end{align*}
    Moreover, for any $c > 0$, there must exist terms $\frac{1}{p_n}, \, \frac{1}{p_m}$ of the price sequence with $N < m < n$ satisfying 
    \begin{itemize}[$\bullet$]
        \item $\sup\{\tfrac{1}{p_n} : n > N\} \leq \frac{1}{p_n} + c$ and 
        \item $\inf\{\tfrac{1}{p_n} : n > N\} \geq \frac{1}{p_n} - c$;
    \end{itemize}
    otherwise, the tail of the sequence would have, respectively, a lesser upper bound and a greater lower one. Suppose we choose $c := \frac{\delta}{2(2 + \varepsilon_\alpha - \varepsilon_\beta)}$.
    
    Then the statements above, along with the definitions of $\lim\inf$ and $\lim\sup$, imply that
    \begin{align*}
        (1 - \varepsilon_\beta)\lim\sup\tfrac{1}{p_n} -  (1 + \varepsilon_\alpha) \lim\inf\tfrac{1}{p_n}
        &\leq (1 - \varepsilon_\beta)\sup\{\tfrac{1}{p_n} : n > N\} -  (1 + \varepsilon_\alpha) \lim\inf\{\tfrac{1}{p_n} : n > N\} \\
        &\leq (1 - \varepsilon_\beta)(\tfrac{1}{p_n} + c) -  (1 + \varepsilon_\alpha) (\tfrac{1}{p_m} - c) \\
        &= (1 - \varepsilon_\beta)\tfrac{1}{p_n} -  (1 + \varepsilon_\alpha) \tfrac{1}{p_m} + c(2 + \varepsilon_\alpha - \varepsilon_\beta ) \\
        &= (1 - \varepsilon_\beta)\tfrac{1}{p_n} -  (1 + \varepsilon_\alpha) \tfrac{1}{p_m} + \tfrac{\delta}{2} \\
        &< \delta.
    \end{align*}
    Therefore, since $-L = (1 - \varepsilon_\beta)\lim\sup\tfrac{1}{p_n} -  (1 + \varepsilon_\alpha) \lim\inf\tfrac{1}{p_n}$ is strictly less than $\delta$ for every $\delta > 0$, we must have $L \geq 0.$ 
    \end{proof}
    To conclude, we note that the price window ensures the rational market price of a stablecoin will be in the range $[1 - \varepsilon_\beta, \, 1 + \varepsilon_\alpha] \subseteq [1 - \varepsilon, \, 1 + \varepsilon]$ for $\varepsilon = \max\{\varepsilon_\alpha, \varepsilon_\beta\}$, unless and until the mechanism's reserves are depleted. This concludes the proof of Theorem~\ref{thm:price-window-impossibility-result_constant-epsilon}.
\end{proof}

\section{Modeling the speculator.}\label{app:speculator}

In this section, we supply a technical description of the speculator. Let $\mathbf{x_t} = (m_t, n_t)$ represent the speculator's portfolio of stablecoins and backing coins, respectively, at time $t \in T$. The speculator interacts with the mechanism $\mathcal{M}$ according to the following interaction model:
\begin{tcolorbox}
    
\subsubsection{The Interaction Model:} At each timestep $t \in T$: 
\begin{enumerate}[label={(\arabic*)}]
    \item The price $p_t$ of the backing coin is revealed. \medskip
    
    \item The speculator either 
    \begin{itemize}
        \item chooses some number $\Delta_t$ of stablecoins to buy ($\Delta_t > 0$) or sell ($\Delta_t < 0$), at the given backing coin price $p_t$, or 
        \item does nothing ($\Delta_t = 0$) and waits until the next timestep.
    \end{itemize}

    The speculator's portfolio is updated to $\mathbf{x}_t = \left(m_{t-1} + \Delta_t, \, n_{t-1} - \frac{\Delta_t}{p_t}\right)$, assuming $\MM$ has sufficient reserves. (If $\MM$'s reserves are insufficient, the speculator receives $R_{t-1} <  \frac{\Delta_t}{p_t}$ backing coins.)

    \item The stablecoin mechanism's reserves change to $R_t = \max\{R_{t-1} + \frac{\Delta_t}{p_t}, 0\}$ in response to the speculator's action.
\end{enumerate}
\end{tcolorbox}

In our instantiation, we assume that the speculator derives intrinsic value from the backing coin and seeks to maximize the number of backing coins it holds. The speculator assigns no intrinsic value to stablecoins, viewing them solely as an instrument for acquiring additional backing coins. Accordingly, the speculator interacts with $\MM$ by buying (minting) and selling (redeeming) stablecoins in exchange for backing coins, with the objective of increasing its backing-coin holdings.

Specifically, at each timestep the speculator aims to maximize its expected future \emph{utility}, where utility is measured as the expected number of backing coins in its portfolio. Suppose the speculator values each stablecoin as equivalent to some quantity $s_1$ of backing coins.\footnote{We define $s_1 = \max_{x \in \mathcal{D}} \big\{ (1-\delta)^\frac{1}{F(x)} \cdot \frac{1}{\mathbb{E}[p \mid p \leq x]} \big\}$, where $F(x)$ denotes the cumulative distribution function for the distribution (known to the speculator) from which backing coin prices are drawn. This formulation captures the speculator's valuation of a stablecoin as the price $x$ that optimally balances its preference for selling at a high price against its willingness to wait for such a price to materialize.}
Then, the speculator's utility (measured in units of the backing coin) for a portfolio $\mathbf{x} = (m, n)$ is given by
\begin{align*}
    u(\mathbf{x}) := \mathbf{x} \cdot \mathbf{s},
\end{align*}
where $\mathbf{s} = (s_1, 1)$ represents the speculator's component-wise valuation of its assets, and $(\cdot)$ denotes the standard dot product.

To model how the speculator chooses the number $\Delta_t$ of stablecoins to buy ($\Delta_t > 0$) or sell ($\Delta_t < 0$), we note that if it chooses to buy or sell $\Delta_t$ stablecoins in timestep $t$, its portfolio becomes 
\begin{equation*}
\mathbf{x_{t-1}} + \mathbf{\Delta_t}(\Delta_t, p_t), \ \ \text{ where } \ \  \mathbf{\Delta_t}(\Delta_t, p_t) := \Delta_t \begin{pmatrix} 1 \\ -1/p_t \end{pmatrix}.
\end{equation*}
(Note: We assume in this section that the mechanism $\mathcal{M}$ does \textit{not} charge transaction fees ($\varepsilon=0$), and so the mint and redeem prices of a stablecoin are both $\mathdollar 1$.) We define as shorthand $\mathbf{\Delta_t^*}(p_t) := \mathbf{\Delta_t}(\Delta_t^*(p_t), p_t)$, where $\Delta_t^*(p_t)$ is the value of $\Delta_t$ which maximizes the speculator's utility at price $p_t$. Since the utility function $u(\mathbf{x_{t-1}} + \mathbf{\Delta_t}(\Delta_t, p_t))$ is linear in $\Delta_t$, the optimal value $\Delta_t^*(p_t)$ is
\begin{itemize}
    \item the highest permissible value of $\Delta_t$ if $\frac{\partial}{\partial\Delta_t} \big(u(\mathbf{x}_{t-1} +\mathbf{\Delta_t}(\Delta_t, p_t)) = s_1 - \frac{1}{p_t} > 0$, 
    \item the lowest permissible value of $\Delta_t$ if $s_1 - \frac{1}{p_t} < 0$, and 
    \item $0$, when $s_1 - \frac{1}{p_t} = 0$. 
\end{itemize}

\paragraph{The feasible interval.} The \emph{permissible} value of $\Delta_t$ above is dictated by the constraints of the problem, as well as the rationality model of the speculator, e.g., its attitude towards risk aversion. For example, since we stipulate that $\mathbf{x}_t \in \mathbb{R}_{\geq 0}^2$ for all $t \in T$, it follows that the permissible values of $\Delta_t$ are in the interval $[-m_{t-1}, \, p_t \, n_{t-1}]$, as the speculator cannot sell more stablecoins than it owns, nor buy more stablecoins than its share of backing coins allows it to buy. Still, the speculator might not be willing to use \emph{all} its funds to act at a given round; while this may be the optimal behavior for a risk-neutral speculator, in practice the speculator would likely only use some fraction of its portfolio at any timestep. To allow for this behavior, rather than change the utility function itself (and thereby lose the analytic advantage of linearity), we externally constrain the size of $\Delta_t$ by defining the \emph{feasible interval} of $\Delta_t$ to be
\begin{equation*}
\mathcal{F}_t = [-(1-\Lambda_1) \, m_{t-1}, \, (1 - \Lambda_2) \, p_t \, n_{t-1}], \ \  \text{ for some fixed constants } \ \  \Lambda_1, \Lambda_2 \in [0, 1].
\end{equation*} 
Larger values of $\Lambda_i$ correlate with a very risk-averse speculator, while smaller values correspond to a more risk-tolerant speculator, with $\Lambda_1 = \Lambda_2 = 0$ representing risk neutrality. (Note that in the main body of the text, we assumed that $\Lambda_1 = \Lambda_2 =: \Lambda$ to simplify notation.)\medskip

\paragraph{The waiting interval.} Working towards a more realistic model for the speculator's behavior, our formulation so far requires the speculator to react, by buying or selling, to the price $p_t$ at every timestep (except in the unlikely case that $s_1 - 1/p_t = 0$). This is because our utility function $u$ only captures the speculator's expected utility at the end of the next timestep; in that case, the speculator would almost always benefit from either buying or selling, and so it always acts. However, the downside to acting (and spending funds) at  every timestep is that if a better price soon comes along, the speculator's ability to act on this better price is limited by its previous action. That is, always reacting to prices, even unexceptional ones, can limit the speculator's ability to take advantage of exceptionally high or low prices. Therefore, the speculator's logic should include some ability to ``look ahead'' and thereby judge for itself whether the current price is worth acting on, or whether it would be better off waiting.

To encode this possibility into our utility function, we assume that the speculator discounts future value at some rate $\delta \in [0, 1)$ (with low $\delta$ corresponding to a patient speculator, and high $\delta$ corresponding to an impatient one). There are many conceivable ways to integrate discounts for future prices into our model, some of which are fairly involved. For our purposes, to make sure our model is amenable to theoretical analyses, we take a relatively simple approach and assume that the speculator's reasoning only extends as far as its next action to buy or sell stablecoins.

Clearly, to reason about future backing coin prices, the speculator's rationality model needs to be equipped with some knowledge about those prices; here we make the standard assumption that this information is in terms of an underlying distribution. In particular, the speculator assumes that the price of the backing coin is drawn independently\footnote{This assumption is crucial to allow us to derive analytical solutions, but we relax it in our simulations and show that the same general patterns of behavior hold in both cases.}
at random according to some distribution $\mathbb{D}$ with probability density function $f(p)$ and cumulative distribution function $F(p) = \int_{-\infty}^p f(x) \, dx$. 
Under these assumptions, the speculator's expected utility from waiting at timestep $t$ and selling $\Delta_{t'}^*(p_{t'})$ stablecoins at some future backing coin price $p_{t'} \leq \frac{1}{s_1}$ is as follows:
\begin{align*}
    u\left(\mathbf{x_{t-1}} + (1-\delta)^{\frac{1}{F(p_{t'})}}\mathbf{\Delta_{t'}^*}(\mathbb{E}[p \, | \, p \leq p_{t'}])\right).
\end{align*}
If $p_t < \frac{1}{s_1}$, then the speculator will not sell stablecoins unless $p_t$ satisfies 

\begin{equation}
\label{eqn:x1}
u\left(\mathbf{x_{t-1}} + \mathbf{\Delta_{t}^*}(p_t)\right) \geq \max_{p_{t'} \leq \frac{1}{s_1}} \left\{ u(\mathbf{x_{t-1}} + (1-\delta)^{\frac{1}{F(p_{t'})}}\mathbf{\Delta_{t'}^*}(\mathbb{E}[p \, | \, p \leq p_{t'}])) \right\};
\end{equation}
 i.e., selling now yields a higher expected utility than what can be expected by waiting. We define $y_1$ to be the value of $p_{t'}$ at which the two expressions are equal.

Similarly, the speculator's expected utility from waiting at timestep $t$ and buying
$\Delta_{t'}^*(p_{t'})$ stablecoins at some future backing coin price $p_{t'} \geq \frac{1}{s_1}$ is 
\begin{align*}
    u\left(\mathbf{x_{t-1}} + (1-\delta)^{\frac{1}{1 - F(p_{t'})}}\mathbf{\Delta_{t'}^*}(\mathbb{E}[p \, | \, p \geq p_{t'}])\right).
\end{align*}
As in the case of selling, if $p_t > \frac{1}{s_1}$, then the speculator will not buy stablecoins unless $p_t$ satisfies 
\begin{equation}
    \label{eqn:x2}
    u\left(\mathbf{x_{t-1}} + \mathbf{\Delta_{t}^*}(p_t)\right) \geq \max_{p_{t'} \geq \frac{1}{s_1}} \left\{ u(\mathbf{x_{t-1}} + (1-\delta)^{\frac{1}{1-F(p_{t'})}}\mathbf{\Delta_{t'}^*}(\mathbb{E}[p \, | \, p \geq p_{t'}])) \right\};
\end{equation} 
i.e., buying now yields a higher expected utility than waiting. Let $y_2$ be the value of $p_{t'}$ at which the two expressions are equal. We refer to $[y_1, y_2]$ as the speculator's \textit{waiting interval} because any prices $p_t$ which fall within the interval cause the speculator to wait, while any values outside this range cause the speculator to act.  The following lemma establishes that the size of the waiting interval is always nonnegative, and identifies the only case for which it can be zero. 

\begin{lemma}
\label{Lemma:interval-no-negative}
   Let $x_1$ and $x_2$ be the prices that maximize the right-hand sides of (\ref{eqn:x1}) and (\ref{eqn:x2}) respectively. The size of the waiting interval is always nonnegative, and is of size zero if and only if $\mathbb{E}[p \, | \, p \geq x_2]= \mathbb{E}[p \, | \, p \leq x_1]$.
\end{lemma}

\begin{proof}%[Proof of \cref{Lemma:interval-no-negative}]
    Let $x_1$ and $x_2$ be the values which maximize the right-hand sides of (\ref{eqn:x1}) and (\ref{eqn:x2}), respectively. To ensure that these values are well-defined, we can simply restrict the range of our distribution to a closed interval. By assumption, $x_1 \leq \frac{1}{s_1} \leq x_2$. Since $y_1$ and $y_2$ are the endpoints of the waiting interval, we can solve for them simply by computing the prices at which the speculator is indifferent between selling stablecoins and waiting (for $y_1$) and buying stablecoins and waiting (for $y_2$).
    
    If the speculator is indifferent between selling and waiting, then $p_t$ satisfies the following, where $x_1$ is the value of $p$ which maximizes the right-hand side:

    \begin{align*}
        u(\mathbf{x_{t-1}} + \mathbf{\Delta_{t}^*}(p_t)) &=   u(\mathbf{x_{t-1}} + (1-\delta)^{\frac{1}{F(x_{1})}}\mathbf{\Delta_{t}^*}(\mathbb{E}[p \, | \, p \leq x_{1}]));
    \end{align*}
    i.e., $y_1 = \bigg( (1-(1-\delta)^\frac{1}{F(x_1)}) \, s_1 + (1-\delta)^\frac{1}{F(x_1)} \frac{1}{\mathbb{E}[p \, | \, p \leq x_1]} \bigg)^{-1}$
    Similarly, we compute that $y_2 = (1-\delta)^{\frac{1}{1 - F(x_2)}} \mathbb{E}[p \, | \, p \geq x_2] + (1 - (1-\delta)^{\frac{1}{1 - F(x_2)}}) \frac{1}{s_1})$. We compute $y_2 - y_1$ as $(y_2 - \frac{1}{s_1}) + (\frac{1}{s_1} - y_1)$ to get:
    \begin{align*}
        y_2 - y_1 &= \bigg(y_2 - \frac{1}{s_1}\bigg) + \bigg(\frac{1}{s_1} - y_1\bigg) \\
        &= \bigg( (1 - \delta)^{\frac{1}{1 - F(x_2)}} \bigg(\mathbb{E}[p \, | \, p \geq x_2] - \frac{1}{s_1} \bigg) \bigg) \\
        &+ \bigg( \frac{y_1}{\mathbb{E}[p \, | \, p \leq x_1]} (1 - \delta)^{\frac{1}{F(x_1)}} \bigg(\frac{1}{s_1} - \mathbb{E}[p \, | \, p \leq x_1] \bigg) \bigg)
    \end{align*}
    Since we define $x_1 \leq 1/s_1 \leq x_2$, both terms are nonnegative, and so their sum can only be zero if both terms are zero. From the expression above, we see that this is only true if $\delta = 1$ or $\mathbb{E}[p \, | \, p \leq x_1] = \mathbb{E}[p \, | \, p \geq x_2] = 1/s_1$. By assumption, $\delta \in [0, 1)$; and so $y_2 - y_1 = 0$ if and only if $\frac{\mathbb{E}[p \, | \, p \geq x_2]}{ \mathbb{E}[p \, | \, p \leq x_1]} = 1$.
\end{proof}

\paragraph{The speculator's behavior.} Based on all the above, the speculator's choice for the quantity $\Delta_t$, therefore, can be captured by the following formula:
\begin{align}
    \Delta_{t}^*(p_{t}) = \begin{cases}
        (1 - \Lambda_2) \, p_t \, n_{t-1}, & p_t > y_2, \\
        0, & y_1 \leq p_t \leq y_2, \\
        -(1 - \Lambda_1) \, m_{t-1}, & p_t < y_1.
    \end{cases}
\end{align}

Note that $u$ is still linear, and hence the speculator at timestep $t$ will either buy or sell as many stablecoins as possible, as long as this number is within the feasible interval, and the price of the backing asset $p_t$ is not within the waiting interval; in the latter case, the speculator will take no action and will reevaluate at timestep $t+1$, with new price $p_{t+1}$.

\section{Computing the speculator's expected portfolio}\label{app:expected-portfolio}

Assume that the price of the backing coin is drawn independently at random from distribution $\mathbb{D}$ at each timestep $t \in T$. We define a ``\emph{round}'' over a sequence of prices as follows: 
suppose the most recent action the speculator took (waiting notwithstanding) was to sell some number of stablecoins. Let $t_a$ represent the next time the speculator \textit{buys} some stablecoins. From here, the speculator, in response to the subsequent prices, might buy some more stablecoins or wait for some more timesteps without buying or selling; but eventually, the speculator will also sell some stablecoins. Let $t_b$ be the first time this occurs, following $t_a$. Let $t_c$ be the last time the speculator sells some stablecoins following $t_b$, before subsequently buying stablecoins. The ``round'' is defined as the sequence of timesteps $(p_t)_{t \in \{t_a, t_{a+1}, \ldots, t_c\}}$. Note that we define a round like so, because using stablecoins to attain more backing coins is a two-step process: The speculator must first buy stablecoins to be able to sell them later, at a profit.

Since the price in each timestep is independent, the speculator's expected profit per round is fixed, as is the expected length of each round. Therefore, to analyze the expected value of the speculator's portfolio over time, we simply need to quantify the speculator's expected profit per round. We first explain how to do that for a single round, and then devise a general formula.

\subsection*{Computing the speculator's expected portfolio after one round}
We first compute the expected change in the speculator's portfolio between $t_{a-1}$ and $t_{b-1}$. During this interval, the speculator is only buying stablecoins (or waiting). Suppose the speculator buys at exactly $i \leq b - a$ timesteps. The impact of a buy operation at price $p \geq y_2$ on the speculator's portfolio can be expressed as $A(p) \, \mathbf{x}$, where  
\begin{equation*}
A(p) := \begin{pmatrix} 1 & (1-\Lambda_1) \, p  \\ 0 & \Lambda_1 \end{pmatrix}.
\end{equation*}
Therefore, the expected value of the speculator's portfolio at $t_{b-1}$ with respect to its original portfolio $\mathbf{x}_{t_{a-1}}$ at time $t_{a-1}$ is as follows:
\begin{align*}
    \mathbb{E}[\mathbf{x}_{t_{b-1}}] &= A(\mathbb{E}[p \, | \, p \geq y_2])^i  \mathbf{x}_{t_{a-1}} = \begin{pmatrix} 1 & (1-\Lambda_1^i) \, \mathbb{E}[p \, | \, p \geq y_2] \\ 0 & \Lambda_1^i \end{pmatrix} \mathbf{x}_{t_{a-1}}.
\end{align*}
Similarly, the impact of selling at some price $p \leq y_1$ on the speculator's portfolio can be written as $B(p) \, \mathbf{x}$, where 
\begin{equation*}
    B(p) = \begin{pmatrix} \Lambda_2 & 0 \\ (1-\Lambda_2) \frac{1}{p} & 1 \end{pmatrix},
\end{equation*}
and so the expected value of the speculator's portfolio at $t_{c}$ after selling stablecoins at exactly $j$ timesteps between $t_b$ and $t_c$ is
\begin{align*}
    \mathbb{E}[\mathbf{x}_{t_{c}}] &= B(\mathbb{E}[p \, | \, p \leq y_1])^j \mathbf{x}_{t_{b-1}} =  \begin{pmatrix} \Lambda_2^j & 0 \\ (1-\Lambda_2^j) \, \frac{1}{\mathbb{E}[p \, | \, p \leq y_1]} & 1 \end{pmatrix} \mathbf{x}_{t_{b-1}}.
\end{align*}
So in total, the expected value of the speculator's portfolio at the end of the round is
\begin{align*}
    \mathbb{E}[\mathbf{x}_{t_c}] &= B(\mathbb{E}[p \, | \, p \leq y_1])^j A(\mathbb{E}[p \, | \, p \geq y_2])^i \mathbf{x}_{t_{a-1}} \\
    &= \begin{pmatrix} \Lambda_2^j & \Lambda_2^j (1 - \Lambda_1^i)  \mathbb{E}[p \, | \, p \geq y_2] \\  (1-\Lambda_2^j) \frac{1}{\mathbb{E}[p \, | \, p \leq y_1]} & \Lambda_1^i + (1-\Lambda_1^i)(1-\Lambda_2^j) \frac{\mathbb{E}[p \, | \, p \geq y_2]}{\mathbb{E}[p \, | \, p \leq y_1]} \end{pmatrix} \mathbf{x}_{t_{a-1}} =: M \mathbf{x}_{t_{a-1}}.
\end{align*}
Thus, we have obtained a recursive formula that describes the expected value of the speculator's portfolio at the end of the round as a function of its portfolio at the start of the round. Below we explore how the eigenvalues and vectors of $M$ can be used to derive a general formula for the expected value of the speculator's portfolio, dependent only on the speculator's initial state $\mathbf{x_0}$. 

\subsection*{A general formula for the speculator's expected portfolio}

Let us define
\begin{align}
    \label{Def:M}
    M = \begin{pmatrix} A & B \\ C & D \end{pmatrix} = \begin{pmatrix} \Lambda_2^j & \Lambda_2^j (1- \Lambda_1^i) \mathbb{E}[p \, | \, p \geq y_2] \\ (1 - \Lambda_2^j) \frac{1}{\mathbb{E}[p \, | \, p \leq y_1]} & \Lambda_1^i + (1-\Lambda_1^i)(1 - \Lambda_2^j) \, Y \end{pmatrix},
\end{align}
where $Y := \frac{\mathbb{E}[p \, | \, p \geq y_2]}{\mathbb{E}[p \, | \, p \leq y_1]}$. \medskip 

First, we prove two technical auxiliary results, \vref{Lemma}{Lemma:LB-on-R} and \vref{Corollary}{Cor:A+D<=2}, which are needed in our subsequent proofs. With those at hand, we can show \vref{Lemma}{Lemma:eigenvalues}, which establishes the existence and values of the eigenvalues and eigenvectors of matrix $M$, and \vref{Lemma}{Lemma:general-formula}, which shows how we can use these values to write a general formula for the speculator's portfolio at the end of round $k$.

\begin{lemma}
    \label{Lemma:LB-on-R}
    Let $M$ be as defined in \vref{Equation}{Def:M}. Then 
    \begin{align}
        \label{Def:R}
        R := \sqrt{(A-D)^2 + 4BC}
    \end{align}
    is lower-bounded by $|A + D - 2|$, 
    with equality holding if and only if any of $Y, \Lambda_1^i, \Lambda_2^j$ equals $1$.
\end{lemma}

\begin{proof}%[Proof of \ref{Lemma:LB-on-R}]
    By the definition of $M$, we note that $BC$ can be rewritten as
     \begin{align}
         \label{Eqn:BC}
         BC = \Lambda_2^j(1-\Lambda_1^i)(1-\Lambda_2^j)Y = \Lambda_2^j(\Lambda_1^i + (1-\Lambda_1^i)(1-\Lambda_2^j)Y) - \Lambda_1^i\Lambda_2^j = AD - \Lambda_1^i\Lambda_2^j,
     \end{align} 
     which implies
     \begin{align*}
         R &= \sqrt{(A-D)^2 + 4BC} = \sqrt{(A-D)^2 + 4AD - 4\Lambda_1^i\Lambda_2^j} = \sqrt{(A+D)^2 - 4\Lambda_1^i\Lambda_2^j}.
     \end{align*}
    The assumption that $Y \geq 1$ allows us to upper bound $\Lambda_1^i \Lambda_2^j$ as follows:
     \begin{align*}
         \Lambda_1^i \Lambda_2^j &= (1-\Lambda_1^i)(1-\Lambda_2^j) - (1 - \Lambda_1^i - \Lambda_2^j) \\
         &\leq (1-\Lambda_1^i)(1-\Lambda_2^j)Y - (1 - \Lambda_1^i - \Lambda_2^j) \\
         &= A + D - 1,
     \end{align*}
     giving us 
     \begin{align*}
         R &= \sqrt{(A+D)^2 - 4\Lambda_1^i\Lambda_2^j} \geq \sqrt{(A+D)^2 - 4(A + D - 1)} = \sqrt{(A+D - 2)^2} \\
         &= |A + D - 2|.
     \end{align*}
     To complete the argument, we observe that the inequality in the expansion of $\Lambda_1^i\Lambda_2^j$ above becomes equality if and only if $(1 - \Lambda_1^i)(1 - \Lambda_2^j) = (1 - \Lambda_1^i)(1 - \Lambda_2^j)Y$, which occurs when either $Y=1$, or at least one of $\Lambda_1^i, \Lambda_2^j = 1$. Thus, when any of these conditions holds, $R = | A + D - 2|$.
 \end{proof}

\begin{corollary}
    \label{Cor:A+D<=2}
    If $R = |A + D - 2|$, then $A + D \leq 2$, with $A + D = 2$ (and hence, $R = 0$) if and only if $\Lambda_1^i = \Lambda_2^j = 1$.
\end{corollary}
\begin{proof}%[Proof of \ref{Cor:A+D<=2}]
    According to \vref{Lemma}{Lemma:LB-on-R}, $R = |A + D - 2|$ if and only if $Y = 1$ or $\Lambda_1^i = 1$ or $\Lambda_2^j = 1$. But if any of these statements is true, then we must have $A + D \leq 2$, since
    \begin{align*}
        A + D = \Lambda_1^i + \Lambda_2^j + (1 - \Lambda_1^i)(1 - \Lambda_2^j)Y = 
        \begin{cases}
            1 + \Lambda_1^i\Lambda_2^j, & Y = 1,
            \\ \Lambda_i^i + \Lambda_2^j, & \Lambda_1^i = 1 \text{ or } \Lambda_2^j = 1,
        \end{cases}
    \end{align*}
    which, in either case, is never greater than 2, and equals 2 precisely when $\Lambda_1^i = \Lambda_2^j = 1$.
\end{proof}

We are now ready to define the following quantities, which will be used in our formula for the speculator's portfolio after $k$ rounds. \vref{Lemma}{Lemma:eigenvalues} below establishes that these are eigenvalues (with respective eigenvectors) of the matrix $M$ as defined in \vref{Equation}{Def:M}.

% \begin{tcolorbox}[colframe=gray,colback=white,boxrule=2pt,arc=.3em,boxsep=2mm, left=-5pt,right=8pt,top=-5pt,bottom=5pt]
\begin{align}
        \label{Def:a1-a2} 
        a_1 &:= \tfrac12 ( A + D + R), & a_2 &:= \tfrac12 ( A + D - R) \\[10pt]
        \label{Def:c1-c2}
        \mathbf{c_1} &:= \frac{1}{R} 
        \begin{pmatrix}
            \tfrac12 (R + (A - D)) \;\; \| \;\; B \\
            C \;\; \| \;\; \tfrac12 (R - (A - D))
        \end{pmatrix}
        \mathbf{x_0}, 
        & \mathbf{c_2} &:= \frac{1}{R} 
        \begin{pmatrix}
            \tfrac12 (R - (A - D)) \;\; \| \;\; - B  \\
            - C \;\; \| \;\; \tfrac12 (R + (A - D)) 
        \end{pmatrix}
        \mathbf{x_0}.
    \end{align}
% \end{tcolorbox}

\begin{lemma}
    \label{Lemma:eigenvalues}
    The matrix $M$ has two distinct, real eigenvalues, $a_1$ and $a_2$, with respective corresponding eigenvectors $\mathbf{c_1}$ and $\mathbf{c_2}$ given by \vreftwo{Equations}{Def:a1-a2}{Def:c1-c2} if and only if $\Lambda_1^i$ and $\Lambda_2^j$ are not both 1. Moreover, if $\Lambda_1^i \Lambda_2^j \neq 0$, then neither eigenvalue is trivial.
\end{lemma}

\begin{proof}%[Proof of \ref{Lemma:eigenvalues}]
    \vref{Equation}{Def:c1-c2} is well-defined if and only if $R > 0$, which, by \vref{Corollary}{Cor:A+D<=2}, is true precisely when $\Lambda_1^i$ and $\Lambda_2^j$ are not both 1. By \vref{Equation}{Def:a1-a2}, this also implies that $a_1, a_2$ are distinct whenever $R \neq 0$; and, moreover, that $a_1$ is nontrivial.\medskip 

    We define an eigenvalue, eigenvector pair $a, \mathbf{c}$ for $M$ as a nontrivial scalar and vector satisfying $M\mathbf{c} = a\mathbf{c}$. We show that the pair $a_1, \mathbf{c_1}$ satisfies this property; the analogous calculation for $a_2, \mathbf{c_2}$ is nearly identical. 
    \begin{small}
    \begin{align*}
        M \mathbf{c_1} &=
        \frac{1}{R}
        \begin{pmatrix}
            A & B \\ C & D
        \end{pmatrix}
        \begin{pmatrix}
            \tfrac12(R + (A - D)) \;\; \| \;\; B  \\
            C \;\; \| \;\; \tfrac12 (R - (A - D)) 
        \end{pmatrix} 
        \mathbf{x_0} \\
        &=
        \frac{1}{2R}
        \begin{pmatrix}
            \tfrac{1}{2}(2A(R + (A - D)) + 4BC) \;\; \| \;\; 2B(A + \tfrac{1}{2}(R - (A - D))) \\
            2C(D + \tfrac{1}{2}(R + (A - D))) \;\; \| \;\; \tfrac{1}{2}(2D(R - (A - D)) + 4BC) \\
        \end{pmatrix} 
        \mathbf{x_0} \\ \\
        &=
        \frac{1}{2R}
        \begin{pmatrix}
            \tfrac{1}{2}(((A+D) + (A-D))(R + (A - D)) + 4BC) \;\; \| \;\;  B(A + D + R) \\
            C(A + D + R) \;\; \| \;\; \tfrac{1}{2}(((A+D) - (A-D))(R - (A - D)) + 4BC) \\
        \end{pmatrix}
        \mathbf{x_0} \\ \\
        &=
        \frac{1}{2R}
        \begin{pmatrix}
            \tfrac{1}{2}((A+D)(R + (A - D)) + 4BC + (A-D)^2 + (A-D)R) \;\;\| \;\;  B(A + D + R) \\
            C(A + D + R) \;\; \| \;\;  \tfrac{1}{2}((A+D)(R - (A - D)) + 4BC + (A-D)^2 - (A-D)R) \\
        \end{pmatrix}
        \mathbf{x_0} \\ \\
        &=
        \frac{1}{2R}
        \begin{pmatrix}
            \tfrac{1}{2}((A+D)(R + (A - D)) + R^2 + (A-D)R) \;\; \| \;\;  B(A + D + R) \\
            C(A + D + R) \;\; \| \;\; \tfrac{1}{2}((A+D)(R - (A - D)) + R^2 - (A-D)R) \\
        \end{pmatrix}
        \mathbf{x_0} \\ \\
        &=
        \frac{1}{2R}
        \begin{pmatrix}
            \tfrac{1}{2}(A+D+R)(R + (A - D)) \;\; \| \;\;  B(A + D + R) \\
            C(A + D + R) \;\; \| \;\; \tfrac{1}{2}((A+D + R)(R - (A - D))) 
        \end{pmatrix}
        \mathbf{x_0} \\ \\
        &=
        \frac{A + D + R}{2R}
        \begin{pmatrix}
            \tfrac{1}{2}(R + (A - D)) \;\; \| \;\;  B \\
            C \;\; \| \;\; \tfrac{1}{2}(R - (A - D)) \\
        \end{pmatrix}
        \mathbf{x_0} \\
        \\
        &= a_1 \mathbf{c_1}.
    \end{align*}
    \end{small}
    Finally, to complete our proof, we still need to argue about the nontriviality of our eigenvalues. We showed earlier that $a_1$ is never trivial, but we claim that $a_2 = 0$ if and only if $\Lambda_1^i$ or $\Lambda_2^j = 0$. If $a_2 = 0$, then by \vref{Equation}{Def:a1-a2}, we must have $R = A + D$. But in \vref{Lemma}{Lemma:LB-on-R}, we showed that
    \begin{align*}
        R = \sqrt{(A+D)^2 - 4\Lambda_1^i\Lambda_2^j},
    \end{align*}
    making it clear that $R = A + D$ if and only $\Lambda_1^i = 0$ or $\Lambda_2^j = 0$. This completes the proof.
\end{proof}

We are now ready to state the main result of the section, which establishes a general formula for the speculator's portfolio. 

\begin{lemma}
    \label{Lemma:general-formula}
    Let $\mathbf{x^k}$ denote the speculator's portfolio at the end of the $k$th round. Then the general formula for $\mathbf{x^k}$, given that $\Lambda_1$ and $\Lambda_2$ are not both 1, is as follows:
\begin{align*}
    \mathbf{x^k}
    = a_1^k \, \mathbf{c_1} + a_2^k \, \mathbf{c_2},
\end{align*}
with $a_1, a_2$ and $\mathbf{c_1}, \mathbf{c_2}$ as defined in \vreftwo{Equations}{Def:a1-a2}{Def:c1-c2}.
\end{lemma}
\begin{proof}
    If $\Lambda_1^i$ and $\Lambda_2^j$ are not both 1, then we know from \vref{Lemma}{Lemma:eigenvalues} that $R > 0$ and, therefore, the equations for $\mathbf{c_1}, \mathbf{c_2}$ are well-defined.
    We proceed by induction.
    \begin{itemize}
        \item Base case: When $k = 0$, we confirm that $\mathbf{c_1} + \mathbf{c_2} = \mathbf{x_0}$, and so the base case holds.
        \item Inductive step:  Now suppose our claim holds for $k-1$; i.e., that $\mathbf{x^{k-1}} = a_1^{k-1} \mathbf{c_1} + a_2^{k-1} \mathbf{c_2}$. Then applying matrix $M$ to this vector gives us
        \begin{align*}
           \mathbf{x^k} 
            &=
            M \mathbf{x^{k-1}}  \\
            &= M\big(a_1^{k-1} \mathbf{c_1} + a_2^{k-1} \mathbf{c_2}\big) \\
            &= a_1^{k-1} (M \mathbf{c_1}) + a_2^{k-1} (M \mathbf{c_2}) \\
            &= a_1^{k-1} (a_1\mathbf{c_1}) + a_2^{k-1} (a_2 \mathbf{c_2}) \\
            &= a_1^{k} \mathbf{c_1} + a_2^{k} \mathbf{c_2},
        \end{align*}
        where the penultimate step comes from the application of \vref{Lemma}{Lemma:eigenvalues} to each eigenvector $\mathbf{c_1}, \mathbf{c_2}$.
    \end{itemize}
\end{proof}

% For whatever reason, \vref doesn't seem to play nice with section titles.
% So rather than debug, this just remains \ref
\section{Miscellaneous items from Section \ref{sec:case-study}}

\subsection{Mapping \texorpdfstring{$k$}{k} to \texorpdfstring{$t$}{t} in our example}\label{app:k-to-t}
To convert from the expected number of rounds to the expected number of timesteps, we simply need to calculate the exponents $i$ and $j$ from \vref{Appendix}{app:expected-portfolio}; by the assumption of price independence, $i$ and $j$ are constant across rounds. Since $i$ and $j$ are the expected waiting times until the price exceeds $y_2$, and falls below $y_1$, respectively, we have in \vref{Example}{ex:lambda=0} that $i = 1/(1 - F(y_2)) \approx 10.057$ and $j = 1/(F(y_1)) \approx 3.6954$. By construction, we defined the first round as starting with a price exceeding $y_2$; so in total, the expected time to reserve depletion is $i + k(i + j) \approx 227$ timesteps. Looking ahead, this is very much consistent with our simulations on real and synthetic price data; indeed, there we find that over 10,000 iterations, the average time-to-depletion experimentally is roughly 224 timesteps (with a standard deviation between trials of roughly 43 timesteps). 

    \subsection{Lemma for the expansion of \texorpdfstring{$Y$}{Y} in the normally-distributed case}\label{app:normal-distr-expansion-of-Y}
\begin{lemma}\label{lemma:normal_xf(x)_expansion}
    When $f(x)$ is the probability density function of the normal distribution with mean $\mu$ and variance $\sigma^2$, $\int_a^b x f(x) dx$ can be expanded as $\mu \int_a^b f(x) \, dx - \sigma^2\big[  f(z) \big]_{z=a}^{z=b}$. 
\end{lemma}

\begin{proof}
Our goal is to expand the definite integral $\int_a^b xf(x)\,dx$, where $f(x) = \frac{1}{\sqrt{2\pi\sigma^2}} e^{-\frac{(x - \mu)^2}{2\sigma^2}}$, into a more workable form. First, we define intermediate functions $g(x) := (x - \mu)^2$ and $h(y) := e^{-\frac{y}{2\sigma^2}}$. By the Fundamental Theorem of Calculus,
\begin{align*}
    \frac{1}{\sqrt{2\pi\sigma^2}} \int_a^b g'(x) h(g(x)) \, dx &= \frac{1}{\sqrt{2\pi\sigma^2}} \int_{g(a)}^{g(b)} h(y) \, dy \\
    &= \frac{1}{\sqrt{2\pi\sigma^2}} \int_{(a - \mu)^2}^{(b - \mu)^2} e^{-\frac{y}{2\sigma^2}} \, dy \\
    &= \frac{1}{\sqrt{2\pi\sigma^2}} \bigg[ -2\sigma^2 e^{-\frac{y}{2\sigma^2}} \bigg]_{(a - \mu)^2}^{(b - \mu)^2} \\
    &=  -2\sigma^2\big[f(z) \big]_{z=a}^{z=b}
\end{align*}
On the other hand, we can also expand the integral $\int_a^b g'(x) f(g(x)) \, dx$ using the linearity property of integration:
\begin{align*}
    \frac{1}{\sqrt{2\pi\sigma^2}} \int_a^b g'(x) h(g(x)) \, dx &= \frac{1}{\sqrt{2\pi\sigma^2}} \int_a^b 2(x - \mu) \, e^{-\frac{(x - \mu)^2}{2\sigma^2}} \, dx \\
    &= \frac{2}{\sqrt{2\pi\sigma^2}} \bigg( \int_a^b x e^{-\frac{(x - \mu)^2}{2\sigma^2}} \, dx - \mu \int_a^b e^{-\frac{(x - \mu)^2}{2\sigma^2}} \, dx \bigg) \\
    &= 2 \bigg( \int_a^b xf(x) \, dx - \mu \int_a^b f(x) \, dx \bigg)
\end{align*}
Combining both equations gives
$$2 \bigg( \int_a^b xf(x) \, dx - \mu \int_a^b f(x) \, dx \bigg) =  -2\sigma^2\big[f(z) \big]_{z=a}^{z=b},$$
or, equivalently, 
$$ \int_a^b xf(x) \, dx = \mu \int_a^b f(x) \, dx - \sigma^2\big[  f(z) \big]_{z=a}^{z=b}.$$
\end{proof}

\subsection{Reserve depletion results}\label{app:speculator-reserve-depletion-results}

\begin{lemma}\label{Lemma:bounding-a1-a2} [Bounding $a_1, a_2,$ and $\mathbf{c_1}$]
    If $\Lambda_1, \Lambda_2 \in [0, 1]$ and $Y \geq 1$, then the eigenvalues of $M$ satisfy the following properties:
        \begin{enumerate}[(1)]
            \item\label{Prop:a1-geq-1}  $a_1 \geq 1$, with equality holding if and only if any of $Y, \Lambda_1^i, \Lambda_2^j$ equals $1$. \medskip
            \item\label{Prop:a2-leq-1} $a_2 \leq 1$, with equality holding if and only if $\Lambda_1^i = \Lambda_2^j = 1$.\medskip
            \item\label{Prop:a2-geq-(-1)} If $a_2 < -1$, then $D > A$. \medskip
    
            Moreover, \medskip

            \item \label{Prop:c1>0} If $\Lambda_1^i$ and $\Lambda_2^j$ are not both 1, then $\mathbf{c_1}$ is positive for any nonnegative, nontrivial initial state vector $(m_0, \, n_0)$.
        \end{enumerate}
    \end{lemma}

\begin{proof}%[Proof of Lemma~\ref{Lemma:bounding-a1-a2}]
    Properties~\ref{Prop:a1-geq-1} and \ref{Prop:a2-leq-1} follow directly from the definitions of $a_1, a_2$ in \vref{Equation}{Def:a1-a2} and \vref{Lemma}{Lemma:LB-on-R}, which states that $R \geq |A + D - 2|$ with equality holding if and only if $Y = 1$ or $\Lambda_1^i = 1$ or $\Lambda_2^j$ = 1. Together, these two statements imply that

    \begin{align*}
        a_1 = \tfrac12 (A + D + R) \geq  \tfrac12 \big(A + D + |A + D - 2| \, \big) = 
        \begin{cases}
            A + D - 1 > 1, & A + D > 2, \\
            \tfrac12 (A + D) = 1, & A + D = 2, \\
            1, & A + D < 2,
        \end{cases}
    \end{align*}
    and 
    \begin{align*}
        a_2 = \tfrac12 (A + D - R) \leq  \tfrac12 \big(A + D - |A + D - 2| \, \big) = 
        \begin{cases}
            1, & A + D > 2, \\
            \tfrac12 (A + D) = 1, & A + D = 2, \\
            A + D - 1 < 1, & A + D < 2.
        \end{cases}
    \end{align*}
    Now we observe that, for both $a_1$ and $a_2$, the first set of inequalities above become equalities only if $R = |A + D - 2|$, which, as we saw in \vref{Corollary}{Cor:A+D<=2}, implies $A + D \leq 2$. Thus, we need only concern ourselves with the second ($A + D = 2$) and third ($A + D < 2$) cases, for in the first case ($A + D > 2$), the eigenvalues $a_1$ and $a_2$ can never equal 1. \medskip
    
    Consequently, it immediately follows that $a_1 = 1$ if and only if $R = |A + D - 2|$; or, equivalently, if and only if $Y = 1$ or $\Lambda_1^i = 1$ or $\Lambda_2^j = 1$.\medskip
    
    For $a_2$, we note that, because of the strict inequality, the third case and $a_2 = 1$ are mutually exclusive. So the only possible way $a_2$ can equal 1 is if $A + D = 2$; or, equivalently, if and only if $R = 0$; which, as we showed in \vref{Corollary}{Cor:A+D<=2}, can only occur if and only if $\Lambda_1^i = \Lambda_2^j = 1$.
 \medskip
    
    \noindent
    This proves Properties~\ref{Prop:a1-geq-1} and \ref{Prop:a2-leq-1}.
    \medskip
    
    \noindent
    For Property~\ref{Prop:a2-geq-(-1)}, we assume that $D \leq A$. Using this assumption along with \vref{Equation}{Eqn:BC}, we have
    \begin{align*}
        BC = AD - \Lambda_1^i \Lambda_2^j \leq AD \leq A^2,
    \end{align*} 
    giving us
    \begin{align*}
        R &= \sqrt{(A-D)^2 + 4BC} \leq \sqrt{(A-D)^2 + 4A^2} \\ &\leq 
        \sqrt{(|A-D| +2A)^2} = \sqrt{(A-D +2A)^2} = 3A - D.
    \end{align*}
    The definition of $a_2$ implies that
    \begin{align*}
        a_2 = \tfrac12 (A + D - R) \geq \tfrac12(A + D - (3A - D)) = D-A \geq -A \geq -1.
    \end{align*}
    Therefore, if $a_2 < -1$, it must be the case that $D > A$.
    
    Finally, we show Property~\ref{Prop:c1>0}: that if $\Lambda_1^i$ and $\Lambda_2^j$ are not both 1 and $(m_0, n_0)$ is nonnegative and nontrivial, then the eigenvector $\mathbf{c_1}$ is also nontrivial. Suppose $\Lambda_2^j < 1$. Then by the definition of $M$ (\ref{Def:M}), $C > 0$ and $A + D \neq 2$ (since $A + D = 2 \iff R = 0 \iff \Lambda_1^i = \Lambda_2^j = 1$, as we showed in \vref{Corollary}{Cor:A+D<=2}.
    The lower bound on $R$ in \vref{Lemma}{Lemma:LB-on-R} further implies that
    $$Z := \tfrac12 (R - (A - D)) \geq \tfrac12(|A + D - 2| - (A - D)).$$
    If $A + D < 2$, then $Z \geq 1 - A = 1 - \Lambda_2^j > 0$, by the assumption that $\Lambda_2^j < 1$. And similarly, if $A + D > 2$, then $Z \geq D - 1 > (2 - A) - 1 = 1 - A > 0$. (We need not consider the case where $A + D = 2$, as explained above.)
    Consequently, the second component of $\mathbf{c_1}$, $\frac{1}{R}(C \, m_0 + \tfrac12(R - (A - D))\, n_0)$, is  strictly greater than zero for $(m_0, n_0)$ nontrivial whenever $\Lambda_2^j < 1$. \medskip 
    
    Now suppose $\Lambda_2^j = 1$, but $\Lambda_1^i \neq 1$. Then $B > 0$, $C = 0$, and $A - D = 1 - \Lambda_1^i > 0$. This implies
    \begin{align*}
        \tfrac12 (R + (A - D)) &= \tfrac12 (\sqrt{(A - D)^2 + 4BC} + (A - D)) \\
        &= \tfrac12 ( |A - D| + (A - D)) \\
        &= A - D > 0.
    \end{align*}
    Thus, the first component of $\mathbf{c_1}$, $\tfrac{1}{R}(\tfrac12 (R + (A - D)) \, m_0 + B \, n_0)$, is strictly greater than zero for nontrivial $(m_0, \, n_0)$ whenever $\Lambda_2^j = 1$ but $\Lambda_1^i \neq 1$.
    
    Thus, $\Lambda_1^i, \Lambda_2^j$ not both 1 implies that $\mathbf{c_1}$ is nonnegative and nonzero if $(m_0, \, n_0)$ is nonnegative and nonzero. Otherwise, if $\Lambda_1^i = \Lambda_2^j = 1$, then $R =0$ and so $\mathbf{c_1}$ is undefined. 
\end{proof}

    \begin{theorem}\label{Theorem}
        Let $(n_k)_{k \in \mathbb{N}}$ denote the expected sequence of the speculator's backing coin holdings, where $n_k$ represents the amount of backing coins the speculator holds after $k$ rounds. Then, $\lim n_k = +\infty$ if and only if none of $Y, \Lambda_1^i, \Lambda_2^j$ equals $1$ and at least one of $m_0, n_0$ is positive (and neither is negative).
    \end{theorem}
    \begin{proof}
        Define
        \begin{align*}
            c_1 &:= \mathbf{c_1} \cdot \mathbf{e_2} = \tfrac{1}{R}(C \, m_0 + \tfrac12 (R - (A - D)) \, n_0), \\
            c_2 &:= \mathbf{c_2} \cdot \mathbf{e_2} = \tfrac{1}{R}(-C \, m_0 + \tfrac12(R + (A - D)) \, n_0),
        \end{align*}
        where $\mathbf{e_2}$ denotes the unit vector $\begin{pmatrix} 0 \\ 1 \end{pmatrix}$. Then by \vref{Lemma}{Lemma:general-formula}, we can write
        \begin{align*}
            n_k = c_1a_1^k+ c_2a_2^k.
        \end{align*}
        We now argue that, when none of $Y, \Lambda_1^i, \Lambda_2^j$ equals 1, we can always construct a sequence $(s_k)_{k \in \mathbb{N}}$ with the property that $s_k \leq n_k$ for all $k \in \mathbb{N}$, and whose limit tends to $+\infty$. We split our analysis into two cases: $|a_2| \leq 1$ and $|a_2| > 1$.
        \begin{itemize}
            \item \textbf{Case 1 ($|a_2| \leq 1$):} Define $s_k := c_1a_1^k - |c_2|$. Then we have
            \begin{align*}
                n_k &= c_1a_1^k + c_2a_2^k \geq c_1a_1^k - |c_2||a_2|^k \geq c_1a_1^k - |c_2| = s_k.
            \end{align*}
            If $Y, \Lambda_1^i, \Lambda_2^j \neq 1$, then Property~\ref{Prop:a1-geq-1} of \vref{Lemma}{Lemma:bounding-a1-a2} implies that $|a_1| > 1$. Moreover, assuming at least one of $m_0, n_0$ is positive (and neither is negative), then $c_1 > 0$ by Property~\ref{Prop:c1>0} of \vref{Lemma}{Lemma:bounding-a1-a2}. Thus, $\lim s_k = + \infty$; and because $s_k \leq n_k$ for all $k \in \mathbb{N}$, $\lim n_k = + \infty$, as well.\bigskip 
            
            \item \textbf{Case 2 ($|a_2| > 1$):} Property~\ref{Prop:a2-leq-1} of \vref{Lemma}{Lemma:bounding-a1-a2} implies that $a_2$ cannot be greater than 1, and so it must be the case that $a_2 < -1$. Define the sequence $(r_k)_{k \in \mathbb{N}}$ so that $r_k := (c_1 - |c_2|) a_1^N$.
            It follows immediately from \vref{Equations}{Def:a1-a2} that $|a_1| > |a_2|$ whenever $R > 0$ (i.e., when $\Lambda_1^i$, $\Lambda_2^j$ are not both 1, as shown in \vref{Corollary}{Cor:A+D<=2}). Therefore, we have
            \begin{align*}
                n_k &= c_1a_1^k + c_2a_2^k \geq c_1 a_1^N - |c_2| |a_2|^N \geq c_1 a_1^N - |c_2| a_1^N = (c_1 - |c_2|) a_1^N = r_k.
            \end{align*}
            Thus, the sequence $(n_k)_{k \in \mathbb{N}}$ is lower-bounded by $(r_k)_{k \in \mathbb{N}}$ when $a_2 < -1$. To complete the proof, we also need to show that $\lim r_k = +\infty$ or equivalently, that $c_1 > |c_2|$ when $a_2 < -1$. Using the formulas for $c_1, c_2$ above and simplifying, we observe that
            \begin{align*}
                c_1 - |c_2| =
                \begin{cases}
                    \frac{1}{R} (2Cm_0 + (D-A)n_0), & c_2 \geq 0 \\
                    n_0, & c_2 \leq 0.
                \end{cases}
            \end{align*}
        \end{itemize}
        Property~\ref{Prop:a2-geq-(-1)} of \vref{Lemma}{Lemma:bounding-a1-a2} implies that $D - A > 0$; and so if $m_0$ and $n_0$ are nonnegative and nontrivial, then $c_1 - |c_2| > 0$. Thus, $(r_k)_{k \in \mathbb{N}}$ (and, by extension, $(n_k)_{k \in \mathbb{N}}$) tends to $+ \infty$.\medskip 
        
        Above, we argued that if none of $Y, \Lambda_1^i, \Lambda_2^j$ equals 1, then $\lim n_k = + \infty$. Now we show that if any of these conditions fails, then $(n_k)_{k \in \mathbb{N}}$ is bounded. By Property~\ref{Prop:a1-geq-1} of \vref{Lemma}{Lemma:bounding-a1-a2}, if any of $Y, \Lambda_1^i, \Lambda_2^j$ equals 1, then $a_1 = 1$. Thus, the term $c_1 a_1^k = c_1$ for all $k \in \mathbb{N}$. And since $|a_1| \geq |a_2|$, we must have $|a_2| \leq 1$, and so $- |c_2| \leq c_2 a_2^k \leq |c_2|$. Thus, $n_k \in \big[c_1 - |c_2|, \, c_1 + |c_2| \big]$ for all $k \in \mathbb{N}$.
    \end{proof}

\subsection{Numerical approximations and simulations}\label{app:simulations}

\subsubsection{Numerical approximations}
In our numerical approximations, we simply use a standard root-finding algorithm to numerically compute the expected value $k$ of reserve depletion, given a speculator with waiting interval $[y_1, y_2]$. The speculator's waiting interval is computed according to the formulas outlined in \vref{Appendix}{app:speculator}; in particular, the speculator comes equipped with patience and risk aversion parameters\footnote{For the sake of simplicity, we assume $\Lambda_1 = \Lambda_2$; as we will see later on, this assumption makes little difference in the overall behavior of reserve depletion.} $\delta$ and $\Lambda$, and computes its optimal waiting interval, assuming prices are drawn independently at random from the normal distribution $\mathcal{D}$; we assume that this distribution is known to the speculator beforehand. We then use our root-finding algorithm on the equation in \vref{Lemma}{Lemma:general-formula} (\vref{Appendix}{app:expected-portfolio}) to solve for the expected number of rounds until reserve depletion.

\begin{figure}[h!]
    \centering
    \includegraphics[width=0.75\linewidth]{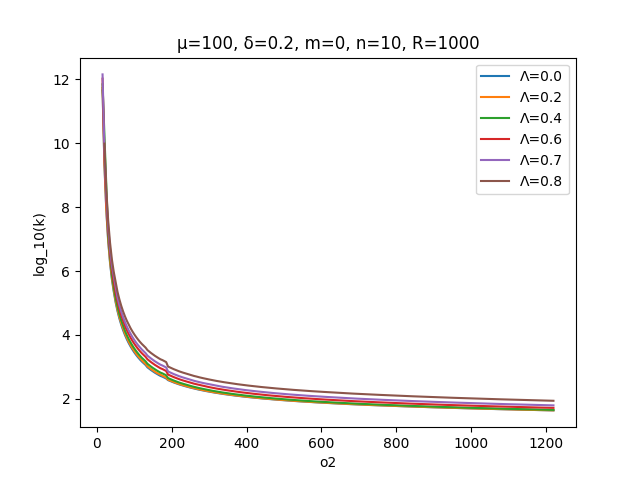}
    \caption{The log of the expected time to depletion for prices drawn independently at random from a normal distribution with $\mu = 100$ and various values of $\sigma^2$, given a speculator with $\delta=0.2$, $\Lambda \in [0, 0.2, 0.4, 0.6, 0.8]$, $n_0 = 10$, and initial reserves $R_0 = 1000$.}
    \label{fig:different-lambdas}
\end{figure}

Below, we describe how each of the speculator's input parameters impacts the expected time to reserve depletion, based on \vref{Lemma}{Lemma:general-formula} in \vref{Appendix}{app:expected-portfolio} and the results of our numerical approximations.
\begin{itemize}
    \item \textbf{The volatility of the backing coin:} The higher the volatility of a price sequence, the faster the rate of reserve depletion. This is clearest to see from the expansion of $Y$ for the independent and normally-distributed case. The same pattern can be seen in our numerical results over a normal distribution in \vref{Figure}{fig:different-deltas}. We note in the figure that the time to depletion increases \textit{dramatically} for small values of $\sigma^2$ -- note that the scale of the $y$-axis is logarithmic. In the next section, we show that the same behavior arises even when we relax the assumption of independent price draws.
    \item \textbf{The speculator's patience} $\delta$\textbf{:} Generally speaking, the lower the value of $\delta$, the larger the speculator's ``waiting interval,'' since more patient speculators are willing to wait longer for ``more exceptional'' prices. The end result is that more patient speculators generally deplete a system's reserves more quickly, as we can see in \vref{Figure}{fig:different-deltas}. Moreover, we see in the figure that the size of $\delta$ has quite a large impact on the rate of depletion, with large values of $\delta$ corresponding to extremely long depletion times. In fact, in our simulations, we could not even compute the time to depletion for $\delta > 0.7$ using standard Python packages without running out of memory in the process.
    \item \textbf{The speculator's risk aversion} $\Lambda$\textbf{:} We proved in \vref{Appendix}{app:speculator-reserve-depletion-results} that the two conditions which prevent reserve depletion are $Y = 1$ or $\Lambda = 1$. In the case of $\Lambda=1$, the speculator is so risk averse that it never opts to buy or sell any stablecoins, and so of course the system's reserves are never depleted -- in fact, they're never changed at all. Unsurprisingly, for $\Lambda < 1$, we find that higher values of $\Lambda$ correspond with longer times-to-depletion. But interestingly, our numerical results did not show a very notable difference in depletion rates for various values of $\Lambda$ (see \vref{Figure}{fig:different-lambdas}), which suggests that assuming $\Lambda = 0$ might not have a great impact on subsequent results. This is convenient because much of the mathematical complexity of the proofs in the appendix stems from nonzero values of $\Lambda$.
    \item \textbf{The ratio} $\tfrac{R_0}{n_0}$\textbf{:} As expected, the higher this ratio (i.e., the more initial funds the system has relative to the speculator), the slower the rate of reserve depletion. This is intuitively obvious, and it is reflected in the formula for $k$ when $\Lambda = 0$. We find that the same pattern holds for general $\Lambda$.
\end{itemize}

\subsubsection{Simulations}
In order to adapt to the new pattern of prices in which the mean and standard deviation are constantly changing, we coded a version of the speculator with a dynamic waiting interval. In particular, we altered the speculator's method of computing its waiting interval $[y_1, \, y_2]$ from the optimization method outlined in \vref{Appendix}{app:speculator} to a more time-sensitive approach. At each time $t$, the speculator has an estimate $\mu_t, \sigma_t$ of the mean and standard deviation of prices over the last few timesteps. The speculator then defines its waiting interval at time $t$ to be $[\mu_t - c \cdot \sigma_t, \mu_t + c\cdot \sigma_t]$ for some constant $c$.\footnote{We found in practice that $c\approx 3.5$ yielded good results.}

\begin{figure}[t]
  \begin{subfigure}{0.48\textwidth}
    \includegraphics[width=\linewidth]{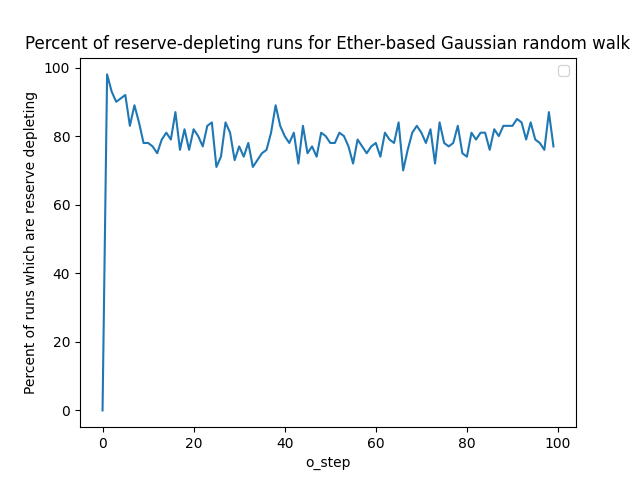}
    \caption{ETH: Percent of runs which were reserve-\\depleting in simulations.} \label{fig:ETH_per}
  \end{subfigure}%
  \hspace*{\fill}   % maximize separation between the subfigures
  \begin{subfigure}{0.48\textwidth}
    \includegraphics[width=\linewidth]{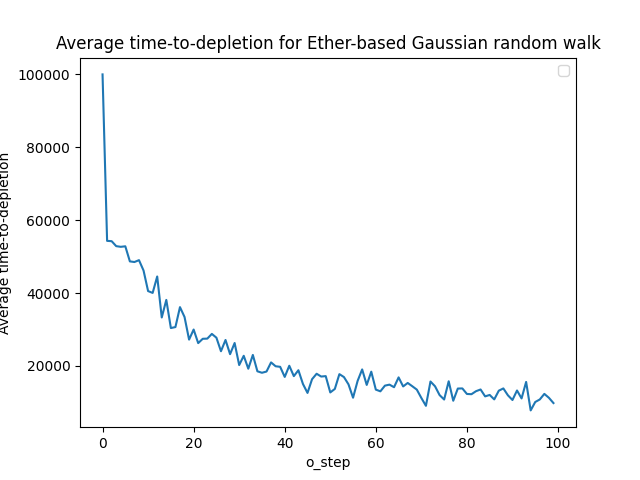}
    \caption{ETH: Average time-to-depletion of the \\reserve-depleting runs.} \label{fig:ETH_ttd}
  \end{subfigure}%
  \hspace*{\fill}   % maximizeseparation between the subfigures
  \\
  \begin{subfigure}{0.48\textwidth}
    \includegraphics[width=\linewidth]{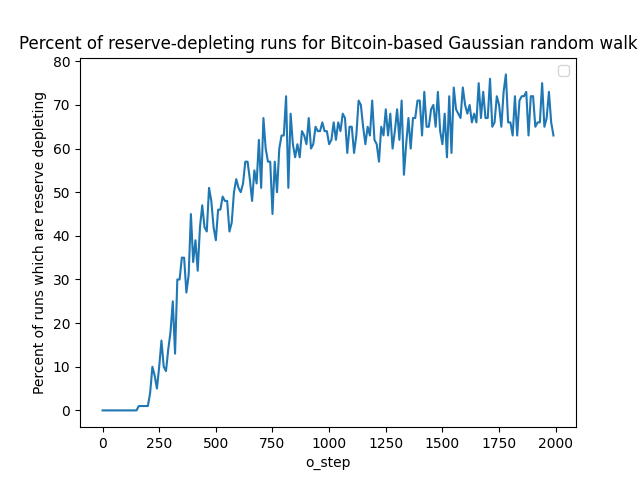}
    \caption{BTC: Percent of runs which were reserve-\\depleting in simulations.} \label{fig:BTC_per}
  \end{subfigure}
  \begin{subfigure}{0.48\textwidth}
    \includegraphics[width=\linewidth]{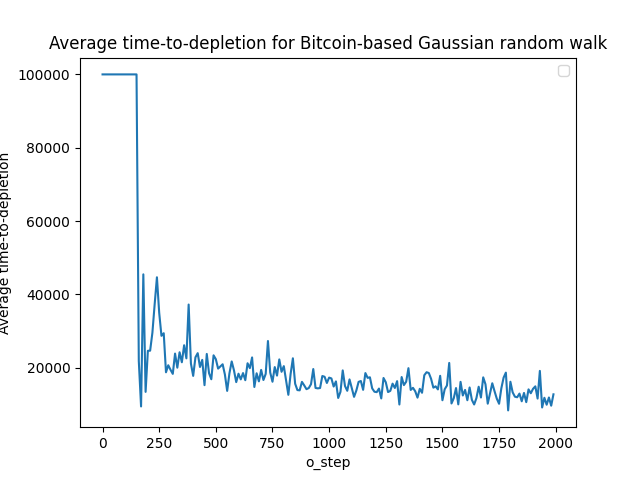}
    \caption{BTC: Average time-to-depletion of the reserve-depleting runs.} \label{fig:BTC_ttd}
  \end{subfigure}
\caption{The results of our Gaussian random walk simulations for both Ethereum and Bitcoin. The mean step-size $\mu_{step}$ and random walk starting point were fixed for each currency throughout, and 100 random walk simulations were run for each value of $\sigma_\text{step}$ for up to 100,000 iterations. The initial reserves contained $R_0 = 1000$ coins, and the speculator was endowed with $n_0 = 100$ coins at the start of each run.} \label{fig:BTC-and-ETH_random_walks}
\end{figure}

\end{document}